%% file: orient.tex
\documentclass[US-letter,USenglish,authorcolumns,autoref]{lipics-v2019}
\input{paper_macros.tex}

\definecolor{winered}{rgb}{0.5,0.2,0}

\usepackage{hyperref}
\hypersetup{ 
    colorlinks = true,
    allcolors={winered},
}

\usepackage{wrapfig}

\usepackage[normalem]{ulem}
\usepackage{changepage}
\DeclareMathOperator{\firstone}{first1}
\DeclareMathOperator{\lastone}{last1}
\DeclareMathOperator{\firstzero}{first0}
\DeclareMathOperator{\lastzero}{last0}

\renewcommand{\tt} {\mathtt}

\newcommand{\fulltree} {\mathrm{T}}

\newcommand{\A} {\mathbf{A}}
\newcommand{\B} {\mathbf{B}}
\newcommand{\N} {\mathbf{N}}
\newcommand{\C} {\mathbf{C}}
\newcommand{\Set} {\mathbf{S}}
\newcommand{\OS} {\mathcal{OS}}

\newcommand{\cycletree} {\mathbb{T}}
\newcommand{\treeroot}[1] {\mathit{r}_n}
\newcommand{\concattree} {\mathcal{T}}

\DeclareMathOperator{\RCL}{RCL}
\DeclareMathOperator{\ap}{ap}

\DeclareMathOperator{\parent}{par}

\nolinenumbers

\title{Construction of orientable sequences in $O(1)$-amortized time per bit}
\titlerunning{~~Orientable sequences}

\author{Daniel Gabri\'{c}}{University of Guelph, Canada}{}{}{}
\author{Joe Sawada}{University of Guelph, Canada}{}{}{}

\authorrunning{~}

\Copyright{D. Gabri\'{c}, J. Sawada}

\begin{document}
\maketitle

\begin{abstract}
%An orientable sequence of order $n$ is a universal cycle for a set $\Set \subseteq \{0,1\}^n$ such that if $\tt{a}_1\tt{a}_2\cdots \tt{a}_n \in \Set$, then its reversal $\tt{a}_n\cdots \tt{a}_2\tt{a}_1 \notin \Set$.  
An orientable sequence of order $n$ is a cyclic binary sequence such that each length-$n$ substring appears at most once \emph{in either direction}.
Maximal length orientable sequences are known only for $n\leq 7$, and a trivial upper bound on their length is $2^{n-1} - 2^{\lfloor(n-1)/2\rfloor}$.  This paper presents the first efficient algorithm to construct orientable sequences that reach this upper bound, asymptotically; more specifically, our algorithm constructs orientable sequences via cycle-joining and a successor-rule approach requiring $O(n)$ time per bit and $O(n)$ space.  
This answers a longstanding open question from Dai, Martin, Robshaw, Wild  [Cryptography and Coding III (1993)].  Applying a recent concatenation-tree framework, the same sequences can be generated in $O(1)$-amortized time per bit using $O(n^2)$ space.
Our sequences are applied to find new longest-known (aperiodic) orientable sequences for $n\leq 20$.

%applying the recent theory of ``concatenation trees'' we present an algorithm that uses $O(1)$-amortized time per bit and $O(n^2)$ space.    Furthermore, we provide an $O(n)$-time successor rule for our constructed orientable sequences. 
\end{abstract}

\newpage

%=====================================================================
\section{Introduction}
%=====================================================================

%This paper is concerned with a special type of cyclic sequence where each $n$-tuple of consecutive bits corresponds to an instance of some combinatorial object. When all instances of a given object appear exactly once as such an $n$-tuple, then the sequence is called a \defo{universal cycle} of span $n$.  The most well studied universal cycles are those for the set of all binary strings, and they are known as \defo{de Bruijn sequences}.  There are wide range of de Bruijn sequence constructions including those derived from LFSRs, 

%A \defo{universal cycle} of span $n$ is a cyclic sequence $\mathcal{U} = u_1u_2\cdots u_N$  with the property that all instances of some combinatorial object are given by the consecutive $n$-tuples $u_1u_2\cdots u_n$, $u_2u_3\cdots u_{n+1}, \ldots, u_Nu_1\cdots u_{n-1}$.  When

%A \defo{universal cycle} of order (span) $n$ for a set  $\Set$ of length $n$ strings is a cyclic string of length $|\Set|$ where every string in  $\Set$ appears exactly once as a substring. 
%
%An \defo{orientable sequence} of order $n$ (an $\OS(n)$) is a universal cycle for a set $\Set \subseteq \{0,1\}^n$ such that if $\tt{a}_1\tt{a}_2\cdots \tt{a}_n \in \Set$, then its reversal $\tt{a}_n\cdots \tt{a}_2\tt{a}_1 \notin \Set$. 

%Reverse and complement?

 Orientable sequences were introduced by Burns and Mitchell~\cite{BM} and studied by Dai, Martin, Robshaw, and Wild~\cite{Dai} with applications related to robotic position sensing.  In particular, consider an autonomous robot with limited sensors.  To determine its location on a cyclic track labeled with black and white squares, the robot scans a window of $n$ squares directly beneath it.  For the position \emph{and} orientation to be uniquely determined, the track must be designed with the property that each length $n$ window can appear at most once in \emph{either direction}.  A cyclic binary sequence (track) with such a property is called an \defo{orientable sequence} of order $n$ (an $\OS(n)$).  By this definition, an orientable sequence does not contain a length-$n$ substring that is a palindrome.

%Consider an autonomous robot with limited sensors. The robot is placed on a cyclic track labeled with white and black squares. The goal is to design a sequence of white and black squares so that the robot can uniquely identify its position and orientation on the track by only scanning the squares in its immediate vicinity. Suppose that, for some $n\geq 1$, the robot can only see a window of $n$ squares around it. Then, any cyclic binary sequence with the property that all length-$n$ substrings occur at most once is enough to uniquely identify position. Orientation can be determined by adding a restriction on such a sequence, namely, if a length-$n$ substring appears, then its reversal does not appear. 

%An \defo{orientable sequence} of order $n$ (an $\OS(n)$) is a cyclic binary sequence such that each length-$n$ substring appears at most once \emph{in either direction}.  
%
%By this definition, an orientable sequence does not contain a length-$n$ substring that is a palindrome.  
%

\begin{exam}
Consider $\mathcal{S}=001011$.
In the forward direction, including the wraparound, $\mathcal{S}$ contains the six 5-tuples $00101$, $01011$, $10110$, $01100$, $11001$, and $10010$; in the reverse direction $\mathcal{S}$ contains $11010$, $10100$, $01001$, $10011$, $00110$, and $01101$.  Since each substring is unique, $\mathcal{S}$ is an $\OS(5)$ with length (period) six. 
\end{exam}

\noindent
Orientable sequences do not exist for $n=1$, and somewhat surprisingly, the maximum length $M_n$ of an $\OS(n)$ is known only for $1 < n \leq 7$. Since the number of palindromes of length $n$ is $2^{\lfloor (n+1)/2\rfloor}$, a trivial upper bound on $M_n$  is $(2^n - 2^{\lfloor (n+1)/2\rfloor})/2= 2^{n-1} - 2^{\lfloor(n-1)/2\rfloor}$. 
%

%Orientable sequences were introduced by Dai, Martin, Robshaw, and Wild~\cite{Dai} with applications related to robotic position and orientation sensing. 
In addition to providing a tighter upper bound, Dai, Martin, Robshaw, and Wild~\cite{Dai}  provide a lower bound $L_n$ on $M_n$ by demonstrating the \emph{existence} of $\OS(n)$s via cycle-joining with length $L_n$ asymptotic to their upper bound. They conclude by stating the following open problem relating to orientable sequences whose lengths (periods) attain the lower bound. See Section~\ref{sec:bounds} for the explicit upper and lower bounds.
%

%\noindent
 %Let $\mu$ denote the M\"{o}bius function defined on positive %integers given by
%\begin{center}
%$\mu(n) = \left\{ \begin{array}{ll}
         %1     \    &\ \  \mbox{if $n=1$,}\\
         %(-1)^j \  &\ \  \mbox{if $n$ is the product of $j$ %distinct primes;} \\
%         0  \      &\ \  \mbox{otherwise.} 
%         \end{array} \right. $
%\end{center}
%

\begin{quote}
\emph{We note that the lower bound on the maximum period was obtained using an existence construction \ldots ~ It is an open problem whether a more practical procedure exists for the construction of orientable sequences that have this asymptotically optimal period.}
\end{quote}

Recently, some progress was made in this direction by Mitchell and Wild~\cite{MW}.  They apply Lempel's lift~\cite{lempel} to obtain an $\OS(n)$ recursively from an $\OS(n{-}1)$.  This construction can generate orientable sequences in $O(1)$-amortized time per bit; however, it requires exponential space, and there is an exponential time delay before the first bit can be output.  Furthermore, they state that their work ``only \emph{partially} answers the question, since the periods/lengths of the sequences produced are not asymptotically optimal.''

\begin{result}
\noindent
{\bf Main result}:
 By developing a parent rule to define a cycle-joining tree, we construct an  $\OS(n)$ of length $L_n$ in $O(n)$ time per bit using $O(n)$ space.  Then, by applying the recent theory of \emph{concatenation trees}~\cite{concat}, the same orientable sequences can be constructed in $O(1)$-amortized time per bit using $O(n^2)$ space.

    %\item By apply the classic cycle-joining technique to create a cycle-joining tree  that has a corresponding successor rule to generate orientable sequences of length $L_n$ in $O(n)$ time per bit using $O(n)$ space.\\
  %  \item We apply a recent result~\cite{concat}, to create a concatenation tree from the cycle-joining tree.  By optimizing its construction, we can generate orientable sequences of length $L_n$ in $O(1)$-amortized time using $O(n)$-space.

\end{result}
\medskip 

\noindent
%A preliminary version of this paper was presented at Combinatorial Pattern Matching~\cite{orientable-conf}.

\noindent
{\bf Outline.}  \footnote{This paper was presented in part at Combinatorial Pattern Matching 2024 (CPM 2024)~\cite{G&S-Orientable:2024}.}
In Section~\ref{sec:prelim}, we present necessary background definitions and notation, a review of the lower bound $L_n$ and upper bound $U_n$ from~\cite{Dai}, and a review of the cycle-joining technique.
In Section~\ref{sec:parent}, we provide a parent rule for constructing a cycle-joining tree composed of ``reverse-disjoint'' cycles corresponding to asymmetric bracelets. This leads to our $O(n)$ time per bit construction of orientable sequences of length $L_n$.
In Section~\ref{sec:periodic}, we present properties of the periodic nodes in our cycle-joining tree and in Section~\ref{sec:children}, we provide an algorithm for determining the children of a given node.
In Section~\ref{sec:concat}, we convert our cycle-joining trees to concatenation trees, which leads to a construction requiring $O(1$)-amortized time per bit.  
In Section~\ref{sec:extend} we discuss the algorithmic techniques used to extend our constructed orientable sequences to find longer ones for $n\leq 20$. Then in Section~\ref{sec:aperiodic}, we apply similar techniques to find some longest known acyclic orientable sequences for $n\leq 20$.
We conclude in Section~\ref{sec:fut} with a summary of our results and directions for future research.  Implementations of our algorithms are available for download at \url{http://debruijnsequence.org/db/orientable}.

\subsection{Related work}  \label{sec:related}

Recall the problem of determining a robot's position and orientation on a track. Suppose now that we allow the track to be non-cyclic. That is, the beginning of the track and the end of the track are not connected. Then the corresponding sequence that allows one to determine orientation and position is called an \defo{acyclic orientable sequence}. One does not consider the substrings in the wraparound for this variation of an orientable sequence. Note that one can always construct an acyclic $\OS(n)$ from a cyclic $\OS(n)$ by taking the cyclic $\OS(n)$ and appending its prefix of length $n{-}1$ to the end. See the paper by Burns and Mitchell~\cite{BM} for more on acyclic orientable sequences, which they call \emph{aperiodic $2$-orientable window sequences}.
Alhakim et al.~\cite{AlhakimOrient} generalize the recursive results of Mitchell and Wild~\cite{MW} to construct orientable sequences over alphabets of size two or greater; they also generalize the upper bound, by Dai et al.~\cite{Dai}, on the length of an orientable sequence.
Rampersad and Shallit~\cite{Rampersad&Shallit:2003} showed that for every alphabet of size two or greater, there is an infinite sequence such that for every sufficiently long substring, the reversal of the substring does not appear in the sequence. Fleischer and Shallit~\cite{Fleischer&Shallit:2019} later reproved the results of the previous paper using theorem-proving software. See~\cite{Currie&Lafrance:2016, Mercas:2017} for more work on sequences avoiding reversals of substrings.

%Two single strands of DNA can bind to each other if some specific conditions are met. The binding of DNA strands allows for the creation of other secondary structures, one of which being the \emph{DNA hairpin}. Hairpin-like structures play a crucial role in certain operations on DNA; for example, insertion and deletion. 

%A string $u$ is said to be \emph{hairpin} if $u = xvy\theta(v)z$ for some words $x,y,z$. A string $u$ is said to be \emph{hairpin-$j$-free} if $u = xvy\theta(v)z$ implies $|v|<k$. For the last two definitions, the function $\theta$ is some anti-morphic involution. In many applications, one needs to construct sets of hairpin-free DNA molecules

%Reversal can be seen as an anti-morphic involution that maps a bit to itself. 

%DNA stuff~\cite{Domaratzki:2006, Kari&etal:2005}

%They show that the set of all hairpin-free strings is finite for $j\leq 4$.

%==========================================================
%==========================================================
%==========================================================
\section{Preliminaries} \label{sec:prelim}

Let $\B(n)$ denote the set of all length-$n$ binary strings. 
Let $\alpha = \tt{a}_1\tt{a}_2\cdots \tt{a}_n \in \B(n)$ and $\beta = \tt{b}_1\tt{b}_2\cdots \tt{b}_m \in \B(m)$ for some $m,n\geq 0$. Throughout this paper, we assume $0<1$ and use lexicographic order when comparing two binary strings. More specifically, we say that $\alpha < \beta$ either if $\alpha$ is a prefix of $\beta$ or if $\tt{a}_i < \tt{b}_i$ for the smallest $i$ such that $\tt{a}_i \neq \tt{b}_i$. 
We say that $\alpha$ is a \defo{rotation} of $\beta$ if $m=n$ and there exist strings $x$ and $y$ such that $\alpha = xy$ and $\beta = yx$. 
The \defo{weight} (density) of a binary string is the number of $1$s in the string. 
Let $\overline{\tt{a}}_i$ denote the complement of bit $\tt{a}_i$. Let $\alpha^R$ denote the reversal $\tt{a}_n\cdots \tt{a}_2\tt{a}_1$ of $\alpha$; $\alpha$ is a \defo{palindrome} if $\alpha = \alpha^R$.
For $j\geq 1$, let $\alpha^j$ denote $j$ copies of $\alpha$ concatenated together.
If $\alpha = \gamma^j$ for some non-empty string $\gamma$ and some $j > 1$, then $\alpha$ is said to be \defo{periodic}\footnote{Periodic strings are are also known as \emph{powers} in the literature. The term \emph{periodic} is sometimes used to denote a string of the form $(\alpha\beta)^i \alpha$ where $\alpha$ is non-empty, $\beta$ is possibly empty, $i\geq 1$, and $\frac{|(\alpha\beta)^i \alpha|}{|\alpha\beta|}\geq 2$. Under this definition, the word $\tt{alfalfa}$ is periodic, but $\tt{bonobo}$ is not.}; otherwise, $\alpha$ is said to be \defo{aperiodic} (or \defo{primitive}). Let $\ap(\alpha)$ denote the shortest string $\gamma$ such that $\alpha = \gamma^t$ for some positive integer $t$; we say $\gamma$ is the \defo{aperiodic prefix} of $\alpha$.  Observe that $\alpha$ is aperiodic if and only if $\ap(\alpha) = \alpha$.

A \defo{necklace class} is an equivalence class of strings under rotation. Let $[\alpha]$ denote the set of strings in $\alpha$'s necklace class.  We say $\alpha$ is a \defo{necklace} if it is the lexicographically smallest string in $[\alpha]$. Let $\tilde{\alpha}$ denote the necklace in $[\alpha]$.     Let $\N(n)$ denote the set of length-$n$ necklaces.
A \defo{bracelet class} is an equivalence class of strings under rotation and reversal; let $\langle \alpha \rangle$ denote the set of strings in $\alpha$'s bracelet class. Thus, $\langle \alpha \rangle = [\alpha] \cup [\alpha^R]$. We say $\alpha$ is a \defo{bracelet} if it is the lexicographically smallest string in $\langle \alpha \rangle$. Note that in general, a bracelet is always a necklace, but a necklace need not be a bracelet. For example, the string $001011$ is both a bracelet and a necklace, but the string $001101$ is a necklace not a bracelet.    %Let $\N(n)$ denote the set of length $n$ necklaces and let $\B(n)$ denote the set of length $n$ bracelets.

A necklace $\alpha$ is \defo{symmetric} if it belongs to the same necklace class as $\alpha^R$, i.e., both $\alpha$ and $\alpha^R$ belong to $[\alpha]$.  By this definition, a symmetric necklace is necessarily a bracelet. If a necklace or bracelet is not symmetric, it is said to be \defo{asymmetric}.
Let $\A(n)$ denote the set of all asymmetric bracelets of length $n$.  
Table~\ref{tab:neck} lists all $60$ necklaces of length $n=9$ partitioned into asymmetric necklace pairs and symmetric necklaces.  The asymmetric necklace pairs belong to the same bracelet class, and the first string in each pair is an asymmetric bracelet.  Thus, $|\A(9)| = 14$.
In general, $|\A(n)|$ is equal to the number of necklaces of length $n$ minus the number of bracelets of length $n$;  for $n=6, 7, \ldots 15$, this sequence of values $|\A(n)|$ is given by $1$, $2$, $6$, $14$, $30$, $62$, $128$, $252$, $495$, $968$ and it corresponds to 
sequence \href{https://oeis.org/A059076}{\underline{A059076}} in The On-Line Encyclopedia of Integer Sequences~\cite{oeis}.  Asymmetric bracelets have been studied previously in the context of efficiently ranking/unranking bracelets~\cite{brace}. %One can test whether a string is an asymmetric bracelet in linear time using linear space; see Theorem~\ref{theorem:braceletTest}.

\begin{lemma}\label{theorem:braceletTest}
    One can determine whether a string $\alpha$ is in $\A(n)$ in $O(n)$ time using $O(n)$ space.
\end{lemma}
\begin{proof}
    A string $\alpha$ will belong to $\A(n)$ if $\alpha$ is a necklace and the necklace of $[\alpha^R]$ is lexicographically larger than $\alpha$. These tests can be computed in $O(n)$ time using $O(n)$ space~\cite{Booth}. 
\end{proof}
%
%A variant of the upcoming Lemma~\ref{lem:pal} appears in a paper by Brlek et al. (see Theorem 4 in~\cite{Brlek&Hamel&Nivat&Reutenauer:2004}). We provide a short proof of a slightly different result. We require that the result be in terms of necklaces, and not aperiodic (i.e., primitive) strings. 
%

%
%\begin{proof}
%If follows from Lemma~\ref{lem:pal} where the only option for $\beta_1 = 0^s$.
%\end{proof}

\begin{table}[ht]
    \centering
\begin{tabular}{c | l }
{\bf Asymmetric necklace pairs} & ~~~~~~~~~~~~ {\bf Symmetric necklaces} \\ \hline

\blue{000001011}~,~000001101  &   000000000             ~~~         0001000\underline{11} ~~~ 00\underline{1110111} \\ 
\blue{000010011}~,~000011001  &  00000000\underline{1} ~~~ 000\underline{101101} ~~~  00\underline{1111111} \\ 
\blue{000010111}~,~000011101  &  0000000\underline{11} ~~~ 000\underline{110011} ~~~  0101010\underline{11} \\  
\blue{000100101}~,~000101001  &  000000\underline{101} ~~~ 000\underline{111111} ~~~  01010\underline{1111} \\ 
\blue{000100111}~,~000111001  &  000000\underline{111} ~~~ 00100100\underline{1} ~~~ 010\underline{111111} \\  
\blue{000101011}~,~000110101  &  00000\underline{1001} ~~~ 00100\underline{1111} ~~~  0110110\underline{11} \\ 
\blue{000101111}~,~000111101  &  00000\underline{1111} ~~~ 0010100\underline{11} ~~~ 0110\underline{11111} \\  
\blue{000110111}~,~000111011  &  0000\underline{10001} ~~~ 00\underline{1010101} ~~~  01110\underline{1111} \\ 
\blue{001001011}~,~001001101  &  0000\underline{10101} ~~~ 00\underline{1011101} ~~~ 0\underline{11111111} \\   
\blue{001010111}~,~001110101  &  0000\underline{11011} ~~~ 001100\underline{111} ~~~  111111111 \\ 
\blue{001011011}~,~001101101  &  0000\underline{11111} ~~~ 00\underline{1101011}  \\ 
\blue{001011111}~,~001111101  &  \\ 
\blue{001101111}~,~001111011  &  \\
\blue{010110111}~,~010111011  &  
\begin{comment}
\blue{000001011}~,~000001101  &  000000000 ~~~ 00000000\underline{1}  \\
\blue{000010011}~,~000011001  &  0000000\underline{11} ~~~ 000000\underline{101}  \\
\blue{000010111}~,~000011101  &  000000\underline{111} ~~~ 00000\underline{1001}  \\
\blue{000100101}~,~000101001  &  00000\underline{1111} ~~~ 0000\underline{10001}  \\
\blue{000100111}~,~000111001  &  0000\underline{10101} ~~~ 0000\underline{11011}  \\
\blue{000101011}~,~000110101  &  0000\underline{11111} ~~~ 0001000\underline{11}  \\
\blue{000101111}~,~000111101  &  000\underline{101101} ~~~ 000\underline{110011}  \\
\blue{000110111}~,~000111011  &  000\underline{111111} ~~~ 00100100\underline{1}  \\
\blue{001001011}~,~001001101  &  00100\underline{1111} ~~~ 0010100\underline{11}  \\
\blue{001010111}~,~001110101  &  00\underline{1010101} ~~~ 00\underline{1011101}  \\
\blue{001011011}~,~001101101  &  001100\underline{111} ~~~ 00\underline{1101011}  \\
\blue{001011111}~,~001111101  &  00\underline{1110111} ~~~ 00\underline{1111111}  \\
\blue{001101111}~,~001111011  &  0101010\underline{11} ~~~ 01010\underline{1111}  \\
\blue{010110111}~,~010111011  &  010\underline{111111} ~~~ 0110110\underline{11}  \\
                        & 0110\underline{11111} ~~~  01110\underline{1111}  \\
                        & 0\underline{11111111} ~~~  111111111  \\
\end{comment}                        
\end{tabular}

    \caption{A listing of all $60$ necklaces in $\N(9)$ partitioned into asymmetric necklace pairs and symmetric necklaces. The first column of the asymmetric necklaces corresponds to the $14$ asymmetric bracelets $\A(9)$.}
    \label{tab:neck}
\end{table}

%\red{perhaps to add}~\cite{Borchert&Rampersad:2015, Kemp:1982}

Lemma~\ref{lem:pal} is considered a folklore result in combinatorics on words; see Theorem 4 in~\cite{Brlek&Hamel&Nivat&Reutenauer:2004} for a variant of the lemma. We provide a short proof for the interested reader.

\begin{lemma}
\label{lem:pal}
A necklace $\alpha$ is symmetric if and only if there exists palindromes $\beta_1$ and $\beta_2$ such that $\alpha = \beta_1 \beta_2$.
\end{lemma}
\begin{proof}
Suppose $\alpha$ is a symmetric necklace. By definition, it is equal to the necklace of $[\alpha^R]$. Thus, there exist strings $\beta_1$ and $\beta_2$ such that $\alpha = \beta_1\beta_2 = (\beta_2\beta_1)^R = \beta_1^R\beta_2^R$. Therefore, $\beta_1=\beta_1^R$ and $\beta_2=\beta_2^R$, which means $\beta_1$ and $\beta_2$ are palindromes. Suppose there exists two palindromes $\beta_1$ and $\beta_2$ such that $\alpha = \beta_1 \beta_2$. Since $\beta_1$ and $\beta_2$ are symmetric, we have that $\alpha^R = (\beta_1\beta_2)^R = \beta_2^R\beta_1^R = \beta_2\beta_1$. So $\alpha$ belongs to the same necklace class as $\alpha^R$ and hence is symmetric.
\end{proof}
\begin{corollary} \label{cor:pal}
If $\alpha = 0^s\beta$ is a symmetric bracelet such that the string $\beta$ begins and ends with $1$ and does not contain $0^s$ as a substring, then $\beta$ is a palindrome. 
\end{corollary}

%==========================================================
%==========================================================
%==========================================================
\subsection{Bounds on $M_n$} \label{sec:bounds}
Dai, Martin, Robshaw, and Wild~\cite{Dai} gave a lower bound $L_n$ and an upper bound $U_n$ on the maximum length $M_n$ of an $\OS(n)$.\footnote{\label{note1} These bounds correspond to $\tilde{L}_n$ and $\tilde{U}_n$, respectively, as they appear in~\cite{Dai}.} The lower bound $L_n$ corresponds to the length of the sequence that results from joining all asymmetric necklaces in a specific way.
Their lower bound $L_n$ is the following, where $\mu$ is the 
M\"{o}bius function:
\begin{equation}\label{eqn:L}
L_n =\sum_{\alpha \in \A(n)}|\ap(\alpha)| = \left( 2^{n-1} -  \frac{1}{2} \sum_{d  |  n} \mu(n/d) \frac{n}{d} H(d) \right), \ \ \text{ where } \ \ H(d) = \frac{1}{2} \sum\limits_{i  |  d} i \left( 2^{\lfloor \frac{i+1}{2} \rfloor} + 2^{\lfloor \frac{i}{2} \rfloor +1} \right).%\tag{*}
\end{equation}
%$$  L_n =\sum_{\alpha \in \A(n)}\ap(\alpha)= \left( 2^{n-1} -  \frac{1}{2} \sum_{d  |  n} \mu(n/d) \frac{n}{d} H(d) \right), \ \ \ \text{ where } \ \ \ H(d) = \frac{1}{2} \sum\limits_{i  |  d} i \left( 2^{\lfloor \frac{i+1}{2} \rfloor} + 2^{\lfloor \frac{i}{2} \rfloor +1} \right).$$
%\begin{remark} \label{rem:L}
%$L_n = \displaystyle{\sum_{\alpha \in \A(n)} \ap(\alpha)}$.
%\end{remark}

\noindent
Their upper bound $U_n$ is the following:$^{\ref{note1}}$

\begin{center}
$U_n= \left\{ \begin{array}{ll} 
         2^{n-1} - \frac{41}{9}2^{\frac{n}{2}-1} + \frac{n}{3} + \frac{16}{9},     \    &\ \  \mbox{if $n \bmod 4 = 0$;}\\
         2^{n-1} - \frac{31}{9}2^{\frac{n-1}{2}} + \frac{n}{3} + \frac{19}{9},     \    &\ \  \mbox{if $n \bmod 4 = 1$;}\\
         2^{n-1} - \frac{41}{9}2^{\frac{n}{2}-1} + \frac{n}{6} + \frac{20}{9},    \    &\ \  \mbox{if $n \bmod 4 = 2$;}\\
         2^{n-1} - \frac{31}{9}2^{\frac{n-1}{2}} + \frac{n}{6} + \frac{43}{18},     \    &\ \  \mbox{if $n \bmod 4 = 3$.}\\
         
         \end{array} \right. $
\end{center}

\noindent
%By applying a combination of techniques to our resulting orientable sequences we provide improved lower bounds for $M_n$, denoted by $L^*_n$ (see Table~\ref{table:bounds}), for $n \leq 20$.
%These sequences are then applied to find new longest \emph{aperiodic} orientable sequences~\cite{BM}.

%
%These bounds are calculated in Table~\ref{table:bounds} for $n$ up to 42 as we search for the ``Answer to the Ultimate Question of Life, the Universe, and Everything''.\footnote{As per Douglas Adam's \emph{The Hitchhiker's Guide to the Galaxy}.}  %Later, our algorithm constructs orientable sequences with length equal to the lower bound $L_n$. 

\noindent
These bounds are calculated in Table~\ref{table:bounds} for $n$ up to 20.  This table also illustrates the length $R_n$ of the $\OS(n)$ produced by the recursive construction by Mitchell and Wild~\cite{MW}, starting from an initial orientable sequence of length 80 for $n=8$. The column labeled $L^*_n$ indicates the longest known orientable sequences we discovered by applying a combination of techniques (discussed in Section~\ref{sec:extend}) to our orientable sequences of length $L_n$.

%============================
%============================
\begin{table}[ht]
\begin{center}
\begin{tabular}{r|r r r r}
$n$ & $R_n$ &  $L_n$  &   $L^{*}_n$  &   $U_n$ \\ \hline
 5   &       -&          0   &      \blue{ {\bf 6}}   &     6 \\  
 6   &        -&         6   &     \blue{ {\bf 16}}   &     17 \\ 
 7   &        -&        14   &     \blue{ {\bf 36}}    &     40 \\ 
 8   &         80 &       48   &     \blue{ 92}  &       96 \\ 
 9   &        161 &       126   &     \blue{ 174}  &       206 \\ 
10   &        322 &       300   &     \blue{  416} &       443 \\ 
11   &        645 &       682   &    \blue{   844} &       918 \\ 
12   &        1290 &      1530   &   \blue{   1844} &       1908 \\ 
13   &        2581 &      3276   &    \blue{  3700} &       3882 \\ 
14   &        5162 &      6916   &    \blue{  7694} &       7905 \\ 
15   &        10325 &     14520   &   \blue{  15394} &       15948 \\ 
16   &        20650 &     29808   &    \blue{ 31483} &       32192 \\ 
17   &        41301 &     61200   &   \blue{  63135}&       64662 \\ 
18   &        82602 &     124368   &   \blue{ 128639} &       129911 \\ 
19   &        165205 &    252434   &   \blue{ 257272} &       260386 \\ 
20   &        330410 &    509220   &   \blue{ 519160} &       521964 \\ 
%21   &        -&   1027208   &        &  1045058 \\ 
%22   &       -&    2063732   &        &  2092493 \\ 
%23   &       -&    4147222   &        &  4187256 \\ 
\end{tabular} \ \ \ \  \ \ \ \ \ \ \  \ \ \ 

% n=9  found by extension of 126
% n=11 found by recursion on 405 odd parity + extending 
% n=13 found by recusion on 1836 + extending 
% n=17 recursion  62826 + extending   
% n=18 found by odd-parity
% n=19  found by recursion on 128552  + some extension
% n=20  odd parity  starting at 13000

\end{center}
    \caption{Lower bounds $R_n, L_n, L^*_n$ and upper bound $U_n$ for $M_n$.}
    \label{table:bounds}
\end{table}
%==========================================================
%==========================================================
\subsection{Cycle joining} \label{sec:cycle}

Given $\Set \subseteq \B(n)$, a \defo{universal cycle} $U$ for $\Set$ is a cyclic sequence of length $|\Set|$ that contains each string in $\Set$ as a substring (exactly once).
Thus, an orientable sequence is a universal cycle.
If $\Set = \B(n)$ then $U$ is known as a \defo{de Bruijn sequence}.
Given a universal cycle $U$ for $\Set$, a \defo{successor rule} for $U$ is a function $f:\Set \rightarrow \{0,1\}$ such that $f(\alpha)$ is the bit following $\alpha$ in $U$.  

%
%and we call a UC-successor a \defo{DB-successor}.  
%A UC-successor can be thought of as a successor rule where the underlying object is a universal cycle. 

%Cycle-joining is perhaps the most fundamental technique applied to construct universal cycles; for some applications, see~\cite{etzion-cutting,Etzion1987,EtzionPSR,fred-nfsr,karyframework,huang,jansen,multi,weakorder}.

Cycle-joining is perhaps the most fundamental technique applied to construct universal cycles.  For instance, see~\cite{Dai}.  If $\Set$ is closed under rotation, then it can be partitioned into necklace classes (cycles); each cycle is disjoint.
Let $\alpha = \tt{a}_1\tt{a}_2\cdots \tt{a}_n$ and $\hat \alpha = \tt{\overline{a}}_1\tt{a}_2\cdots \tt{a}_n$; we say $(\alpha, \hat \alpha)$ is a \defo{conjugate pair}.
Two disjoint cycles can be joined if they each contain one string of a \defo{conjugate pair} as a substring.  This approach resembles pioneering algorithm of Hierholzer (1873) to construct an Euler cycle in an Eulerian graph.
\begin{exam} 
%We use the centre dot $\cdot$ to denote string concatenation.
Consider disjoint subsets $\Set_1 = [011111] \cup [001111]$ and $\Set_2 = [010111] \cup [010101]$, where $n=6$.  Then
$U_1 = 110011\underline{110111}$ is a universal cycle for $\Set_1$ and $U_2 = 01\underline{010111}$ is a universal cycle for $\Set_2$.  Since $(110111,010111)$ is a conjugate pair, $U = 110011\underline{110111}01\underline{010111}$ is a universal cycle for $\Set_1 \cup \Set_2$.
\end{exam}
%

%\red{define cycle-joining tree in a clearer way. maybe say something like, a cycle-joining tree is a tree with nodes representing disjoint universal cycles and that two nodes that share an edge must also share a conjugate pair}

\noindent
A \defo{cycle-joining tree} is a tree with nodes representing disjoint universal cycles; an edge between two nodes implies they each contain one string of a conjugate pair.  If $\Set$ is the set of all length-$n$ strings belonging
to the disjoint cycles of a cycle-joining tree, then 
%If  cycles can be joined via conjugate pairs to form a cycle-joining tree, then 
the tree defines a universal $U$ for $\Set$ along with a corresponding successor rule; see Section~\ref{sec:parent} for an example.  
%This strategy was applied by~\cite{Dai} to demonstrate the lower bound $L_n$.
%however, they were never able to efficiently articulate a specific corresponding set $\Set$.  Rather, only the existence of such a set $\Set$ with a corresponding cycle-joining tree.
%
For most universal cycle constructions, a corresponding cycle-joining tree can be defined by a rather simple \defo{parent rule}.  For example, when $\Set = \B(n)$, the following are perhaps the \emph{simplest} parent rules that define how to construct cycle-joining trees with nodes corresponding to necklace cycles represented by $\N(n)$~\cite{binframework,karyframework,concat}.

\begin{itemize}
\item \bblue{Last-$0$}: rooted at $1^n$ and the parent of every other node $\alpha \in \N(n)$ is obtained by flipping the \blue{last $0$}.
\item \bblue{First-$1$}: rooted at $0^n$ and the parent of every other node $\alpha \in \N(n)$ is obtained by flipping the \blue{first $1$}.
\item \bblue{Last-$1$}: rooted at $0^n$ and the parent of every other node $\alpha \in \N(n)$ is obtained by flipping the \blue{last $1$}.
\item \bblue{First-$0$}:  rooted at $1^n$ and the parent of every other node $\alpha \in \N(n)$ is obtained by flipping the \blue{first $0$}.
\end{itemize}

\noindent
These rules induce the cycle-joining trees
$\fulltree_1$, $\fulltree_2$, $\fulltree_3$, $\fulltree_4$ illustrated in Figure~\ref{fig:big4} for $n=6$.  Note that for $\fulltree_3$  and $\fulltree_4$, the parent of a node $\alpha$ is obtained by first flipping the highlighted bit and then rotating the string to its lexicographically least rotation to obtain a necklace.
Each node $\alpha$ and its parent $\beta$ are joined by a conjugate pair, where the highlighted bit in $\alpha$ is the first bit in one of the conjugates.  For example, the nodes $\alpha = 0\rred{1}1011$ and $\beta = 001011$ in $\fulltree_2$ from Figure~\ref{fig:big4} are joined by the conjugate pair $(\rred{1}10110, 010110)$.  
%
%\red{Is the following example necessary?}
%The following successor rule $f$ is for the de Bruijn sequence derived from $\fulltree_3$ (Last 1)~\cite{binframework}:

%
%===================
\begin{figure}[ht]
\centering
\resizebox{5.4in}{!}{\includegraphics{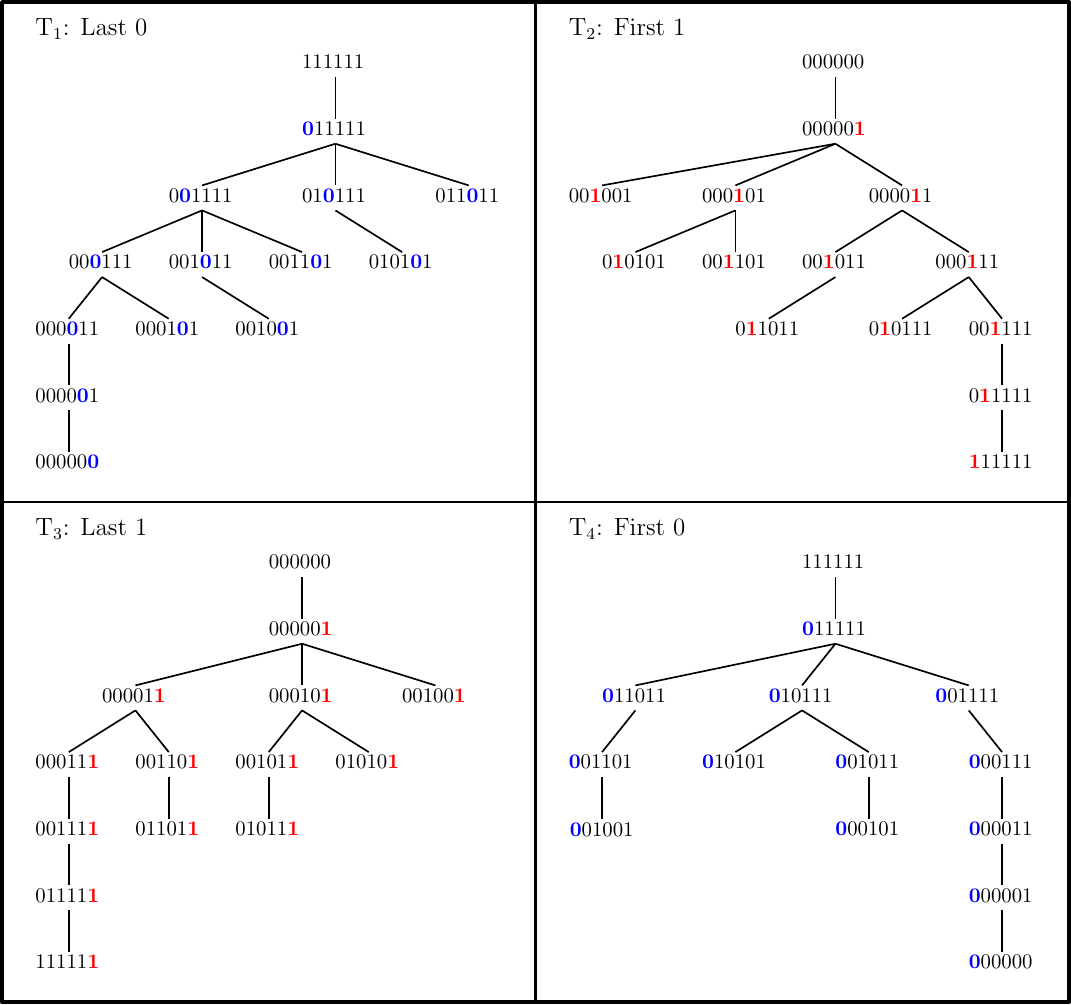}}
\caption{
Cycle-joining trees for $\B(6)$ from simple parent rules. }  
\label{fig:big4}
\end{figure}
%===================

\begin{comment}
%===================
\begin{result}  
\noindent   
Let $\alpha = \tt{a}_1\tt{a}_2\cdots \tt{a}_n$ and let $\gamma = \tt{a}_2 \tt{a}_3 \cdots \tt{a}_{n}1$.  

\begin{center}
$f(\alpha) = \left\{ \begin{array}{ll}
         \overline{\tt{a}}_1 &\ \  \mbox{if $\gamma$ is a necklace;}\\
         {\tt{a}_1} \  &\ \  \mbox{otherwise.}\end{array} \right. $
\end{center}

\vspace{-0.2in}
\end{result}
%===================

\noindent
The universal cycle (de Bruijn sequence) for $n=6$ generated by this successor rule starting at $000000$ is:
\[
0000001111110111100111000110110100110000101110101100101010001001.
\]

\end{comment}

%==========================================================
%==========================================================
%==========================================================
\section{An efficient cycle-joining construction of orientable sequences}  \label{sec:parent}

Consider the set of asymmetric bracelets $\A(n) = \{\alpha_1, \alpha_2, \ldots, \alpha_t\}$.   Recall, that each symmetric bracelet is a necklace.
Let $\Set(n) =   [\alpha_1]  \cup [\alpha_2] \cup \cdots \cup [ \alpha_t ]$.
From~\cite{Dai}, we have $|\Set(n)| = L_n$.  By definition, there is no string $\alpha \in \Set(n)$ such that $\alpha^R \in \Set(n)$.  Thus, a universal cycle for $\Set(n)$ is an $\OS(n)$.  %For the rest of this section, assume $n \geq 8$.  

To construct a cycle-joining tree with nodes $\A(n)$, we apply a combination of three of the four simple parent rules described in the previous section.  First, we demonstrate that there is no such parent rule using at most two rules in combination. 
Assume $n \geq 8$. Observe that none of the necklaces in $\A(n)$ have weight $0$, $1$, $2$, $n{-}2$, $n{-}1$, or, $n$. Thus, 
$0^{n-4}1011$ and $0^{n-5}10011$ are both necklaces in $\A(n)$ with minimal weight three.  Similarly, $00101^{n-4}$ and $001101^{n-5}$ are necklaces in $\A(n)$ with maximal weight $n{-}3$. Therefore, 
when considering a parent rule for a cycle-joining tree with nodes $\A(n)$, the rule must be able to flip a $0$ to a $1$, or a $1$ to a $0$, i.e., if the rule applies a combination of the four rules from Section~\ref{sec:cycle}, it must include one of First-$0$ or Last-$0$, and one of First-$1$ and Last-$1$.

Let $\alpha = \tt{a}_1\tt{a}_2\cdots \tt{a}_n$ denote a necklace in $\A(n)$; it must begin with $0$ and end with $1$.  Then let
\begin{itemize}
    \item $\firstone(\alpha)$ be the necklace $\tt{a}_1\cdots\tt{a}_{i-1}\bblue{0}\tt{a}_{i+1}\cdots \tt{a}_n$, where $i$ is the index of the first $1$ in $\alpha$;
    \item $\lastone(\alpha)$ be the necklace of $[\tt{a}_1\tt{a}_2 \cdots \tt{a}_{n-1}\bblue{0}]$;  
    \item $\firstzero(\alpha)$ be the necklace of $[\rred{1}\tt{a}_2 \cdots \tt{a}_{n}]$;
    \item $\lastzero(\alpha)$ be the necklace $\tt{a}_1\cdots\tt{a}_{j-1}\rred{1}\tt{a}_{j+1}\cdots \tt{a}_n$, where $j$ is the index of the last $0$ in $\alpha$.
\end{itemize}
Note that $\firstone(\alpha)$ and $\lastzero(\alpha)$ are necklaces (easily observed by definition) obtained by flipping the $i$-th and $j$-th bit in $\alpha$, respectively;  $\lastone(\alpha)$ and $\firstzero(\alpha)$ are the result of flipping a bit and rotating the resulting string to obtain a necklace.   The following remark follows from the definition of necklace. 
\begin{remark} \label{rem:lastone}
Let $\alpha=\beta 10^t1$ be a necklace where $\beta$ is some string, and $t \geq 0$.  Then $\lastone(\alpha) =
0^{t+1}\beta1 $.
\end{remark}
Proposition~\ref{prop:2rules} illustrates that for $n$ sufficiently large, no two of the above four parent rules can be applied in combination to obtain a cycle-joining tree with nodes $\A(n)$.

\begin{proposition}~\label{prop:2rules}
Let $\mathrm{p}$ be a parent rule that applies some combination of $\firstone$, $\lastone$, $\firstzero$, and $\lastzero$ to construct a cycle-joining tree with nodes $\A(n)$. Then $\mathrm{p}$ must apply at least three of these rules for all $n\geq 10$.
\end{proposition}
\begin{proof}
Suppose $n\geq 10$. By our earlier observation, any parent rule for a cycle-joining tree with nodes $\A(n)$ must be able to flip a $0$ to a $1$, and a $1$ to a $0$. Therefore, $\mathrm{p}$ must include one of $\firstzero$ or $\lastzero$, and one of $\firstone$ and $\lastone$. 

Suppose $\mathrm{p}$ does not apply $\firstzero$. Then it must apply $\lastzero$. Consider 
three asymmetric bracelets in $\A(n)$: $\alpha_1 = 0^{n-4}1011$, $\alpha_2 = 0^{n-5}10111$, and $\alpha_3 = 0^{n-6}110111$.  Clearly, $\firstone(\alpha_1)=0^{n-2}11$, $\lastone(\alpha_1)=0^{n-3}101$, and $\lastzero(\alpha_1)=0^{n-4}1111$ are symmetric.  Thus, $\alpha_1$ must be the root.  Both
$\firstone(\alpha_2)=0^{n-3}111$ and $\lastzero(\alpha_2)=0^{n-5}11111$ are symmetric, so $\mathrm{p}$ must apply $\lastone$. Both $\lastzero(\alpha_3)=0^{n-6}111111$ and $\lastone(\alpha_3)=0^{n-5}11011$ are symmetric, so $\mathrm{p}$ must apply $\firstone$. 

Suppose $\mathrm{p}$ does not apply $\lastzero$. Then it must apply $\firstzero$. Let $m = 0$ if $n$ is even, and $m=1$ otherwise. Let $\ell = (n-6-m)/2$. Note that $\ell \geq 2$ for $n\geq 10$.  Consider three asymmetric bracelets in $\A(n)$: $\beta_1 = 00101^{n-7}011$, $\beta_2 = 00101^{n-4}$, and $\beta_3 = 0^{\ell+1}10^{m+1}10^\ell11$.  Clearly, $\lastone(\beta_1)=000101^{n-7}01$ is symmetric and $\firstone(\beta_1)=00001^{n-7}011$ is not a bracelet. Additionally, $\firstzero(\beta_1)=0101^{n-7}0111$ is symmetric when $n=10$ and is not a bracelet for all $n>10$.  Thus, $\beta_1$ must be the root. Both $\firstone(\beta_2)=00001^{n-4}$ and $\firstzero(\beta_2)=0101^{n-3}$ are symmetric, so 
 $\mathrm{p}$ must apply $\lastone$. Now for $\beta_3$, we have that $\firstzero(\beta_3)=0^{\ell}10^{m+1}10^\ell 111$ is symmetric. We also have that $\lastone(\beta_3)=0^{\ell+2}10^{m+1}10^\ell1$ is symmetric when $n=11$ and is not a bracelet when $n=10$ or $n>11$. Thus, $\mathrm{p}$ must apply $\firstone$. 
\end{proof}

\noindent
For $n \geq 6$, we choose the lexicographically smallest length-$n$ asymmetric bracelet $\treeroot{n}=0^{n-4}1011$ to be the root of our cycle-joining tree.

%=========================================
\begin{result}
\noindent  \small
{\bf Parent rule for cycle-joining $\A(n)$:} 
Let $\treeroot{n}$ be the root.  Let $\alpha$ denote a non-root node in $\A(n)$. 
Then 

\begin{equation} \label{eq:par}
\parent(\alpha) =  
\left\{ \begin{array}{ll}
         \firstone(\alpha) & \ \ \mbox{if $\firstone(\alpha) \in \A(n)$;}\\
         \lastone(\alpha) & \ \  \mbox{if $\firstone(\alpha) \notin \A(n)$ and $\lastone(\alpha) \in \A(n)$;  }\\
         \lastzero(\alpha) \  &\ \  \mbox{otherwise.}\end{array} \right.
\end{equation}        

\vspace{-0.10in}

 \end{result}
%=========================================

%=========================
%=========================
\begin{theorem}  \label{thm:cycle}
   For $n \geq 6$, the parent rule $\parent(\alpha)$ in \emph{(\ref{eq:par})} induces a cycle-joining tree with nodes $\A(n)$ rooted at $\treeroot{n}$.  The tree has height less than $2(n-4)$.
\end{theorem}

\noindent
Let $\cycletree_n$ denote the cycle-joining tree with nodes $\A(n)$ induced by the parent rule in (\ref{eq:par});   Figure~\ref{fig:cycle9} illustrates $\cycletree_9$. 
The proof of Theorem~\ref{thm:cycle} relies on the following lemma.

%========================================
\begin{figure}[ht]
    \centering
    \resizebox{4.5in}{!}{\includegraphics{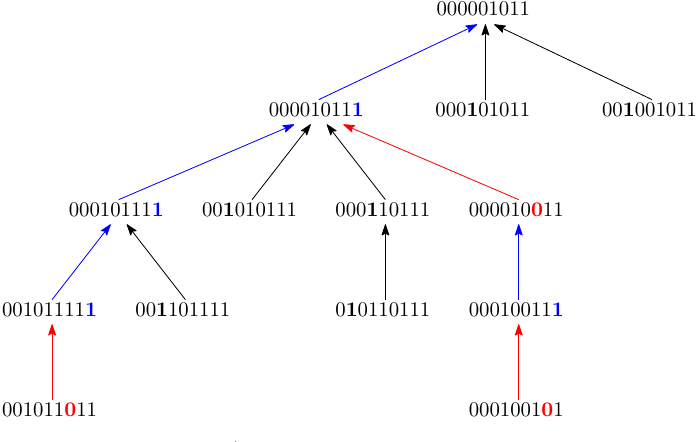}}
    \caption{The cycle-joining tree $\cycletree_9$. The black edges indicate that $\parent(\alpha) = \firstone(\alpha)$; the blue edges indicate that $\parent(\alpha) = \lastone(\alpha)$; the red edges indicate that $\parent(\alpha) = \lastzero(\alpha)$.}
    \label{fig:cycle9}
\end{figure}
%========================================

\begin{comment}
%
\begin{exam}
Consider the asymmetric bracelets $010110111$, $000100111$,  and $001011011$ from Figure~\ref{fig:cycle9}. 
\begin{itemize}
\item $\parent(010110111) = \firstone(0\mathbf{1}0110111) = 000110111$.  
\item  $\parent(000100111) = \lastone(00010011\bblue{1}) = 000010011$, since
$\firstone(000\mathbf{1}00111) = 000000111$ is symmetric.
\item $\parent(001011011) = \lastzero(001011\rred{0}11) = 001011111$, since both $\firstone(00\mathbf{1}011011) = 000011011$ and $\lastone(00101101\bblue{1}) = 000101101$ are symmetric.
\end{itemize}
\end{exam}
\end{comment}

\begin{lemma}  \label{lem:parents}
Let $\alpha \neq \treeroot{n}$ be an asymmetric bracelet in $\A(n)$.  
If neither $\firstone(\alpha)$ nor $\lastone(\alpha)$ are in $\A(n)$, then
the last $0$ in $\alpha$ is at index $n{-}2$ or $n{-}1$, and both $\lastzero(\alpha)$ and $\lastone(\lastzero(\alpha))$ are in $\A(n)$.  
\end{lemma}
\begin{proof}
Since $\alpha$ is an asymmetric bracelet, it must have the form $\alpha = 0^i1\beta 01^j$ where $i,j \geq 1$ and $\beta 0$ does not contain $0^{i+1}$ as a substring.  Furthermore, $1\beta 01^j < (1 \beta01^j)^R$, which implies 
$\beta 01^{j-1} < (\beta 01^{j-1})^R$.

\noindent
{\bf Suppose $j > 2$.} Since $\lastone(\alpha) = 0^{i+1}1 \beta 01^{j-1}$ is not an asymmetric bracelet, we have $1 \beta 01^{j-1} \geq (1 \beta 01^{j-1})^R$.  
Thus, $\beta$ begins with $1$.  
Since $\firstone(\alpha) = 0^{i+1}\beta 01^j$ is not an asymmetric bracelet, Lemma~\ref{lem:pal} implies $\beta 01^j \geq (\beta 01^j)^R$, contradicting the earlier observation that $\beta 01^{j-1} < (\beta 01^{j-1})^R$. Thus, the last $0$ in $\alpha$ is at index $n{-}2$ or $n{-}1$.

\noindent
{\bf Suppose $j=1$ or $j=2$.}  Then the last $0$ in $\alpha$ must be at position $n{-}2$ or $n{-}1$. Write $\alpha = x0y$ where $y=1$ or $y=11$. Since $\alpha$ is a bracelet, it is straightforward to see that $\lastzero(\alpha)=x1y$ is also a bracelet.  If it is symmetric,  Lemma~\ref{lem:pal} implies there exist palindromes $\beta_1$ and $\beta_2$  such that 
 $\lastzero(\alpha)=x1y=\beta_1\beta_2$.  However, flipping the $1$ in $x1y$ that allows us to obtain $\alpha$ implies that $\alpha$ is greater than or equal to the necklace in $[\alpha^R]$, contradicting the assumption that $\alpha$ is an
asymmetric bracelet. Thus, $\lastzero(\alpha)$ is an asymmetric bracelet.  

Consider $\lastone(\lastzero(\alpha))
 = 0^{i+1}1\beta 1^j$.   Let
 $\beta = \tt{b}_1\tt{b}_2\cdots \tt{b}_m$.  
Suppose that $m=0$. Then $\lastone(\lastzero(\alpha))
 = 0^{i+1} 1^{j+1} \Rightarrow \lastzero(\alpha) = 0^{i}1^{j+2}$. Since $j=1$ or $j=2$, we have that $\lastzero(\alpha) = 0^i 111$ or $\lastzero(\alpha) = 0^i1111$. Now $\alpha$ is the result of flipping one of the $1$s in $\lastzero(\alpha)$ to a $0$ and performing the appropriate rotation. But in every case, we end up with $\alpha$ being a symmetric necklace, a contradiction.  Thus, assume $m \geq 1$. 
Suppose $\beta = 1^m$. Then, $\alpha$ is not an asymmetric bracelet, a contradiction.
Suppose $\beta = 0^m$.  If $j=1$, then $\alpha$ is symmetric, a contradiction; if $j=2$, then $\lastone(\lastzero(\alpha)) = 0^{i+1}10^{m}11$ which is in $\A(n)$.
For all other cases, $\beta $ contains at least one $1$ and at least one $0$; $m \geq 2$.
 Since $\beta$ does not contain $0^{i+1}$ as a substring, by Lemma~\ref{lem:pal}, we must show that \blue{(i) $\beta 1^{j-1}$ is less than its reversal  $1^{j-1} \beta^R$},    
recalling that \blue{(ii) $\beta 01^{j-1}$ is less than its reversal $1^{j-1}0\beta^R$}.  Let $\ell$ be the largest index of $\beta$ such that $\tt{b}_{\ell} = 1$. Then $\tt{b}_{\ell+1}\cdots \tt{b}_m=0^{m-\ell}$; note that $\tt{b}_{\ell+1}\cdots \tt{b}_m$ is the empty string when $\ell=m$.
Suppose $j=1$.  From \blue{(ii)}, we have $\tt{b}_1 = 0$ and $\tt{b}_2\cdots \tt{b}_{\ell-1} 1 0^{m-\ell} < 0^{m-\ell}1 \tt{b}_{\ell-1}\cdots \tt{b}_2$. But this implies that $\tt{b}_2\cdots \tt{b}_{m-\ell+1} = 0^{m-\ell}$. Therefore, we have $\beta = 0^{m-\ell+1} \tt{b}_{m-\ell+2} \cdots \tt{b}_m < 0^{m-\ell}1\tt{b}_{\ell-1}\cdots \tt{b}_1 = \beta^R$, hence \blue{(i)} is satisfied. 
Suppose $j=2$.  If $\tt{b}_1 = 0$, then \blue{(i)} is satisfied. Otherwise $\tt{b}_1=1$ and from \blue{(ii)} $\tt{b}_2=0$. From \blue{(ii)}, we get that $\tt{b}_3 \cdots  \tt{b}_{\ell-1}1 0^{m-\ell} < 0^{m-\ell}\tt{b}_{\ell-1}\cdots \tt{b}_3 $. This inequality implies that $\tt{b}_3 \cdots \tt{b}_{m-\ell+2} = 0^{m-\ell}$. Therefore, we have $\beta 1 = 10^{m-\ell+1} \tt{b}_{m-\ell+3}\cdots \tt{b}_m1 < 10^{m-\ell} 1\tt{b}_{\ell-1}\cdots \tt{b}_1 = 1\beta^R$, hence \blue{(i)} is satisfied.
 Thus, $\lastone(\lastzero(\alpha))$ is an asymmetric bracelet.
\end{proof}

\noindent
{\bf Proof of Theorem~\ref{thm:cycle}.}
Let $\alpha$ be an asymmetric bracelet in $\A(n) \setminus \{\treeroot{n}\}$. We can write $\alpha$ as $0^i1\beta$ for some string $\beta$ and $i \geq 1$.
We demonstrate that the parent rule $\parent$ from (\ref{eq:par}) induces a path from $\alpha$ to $\treeroot{n}$, i.e., there exists an integer $j$ such that $\parent^j(\alpha) = \treeroot{n}$.  Note that $\treeroot{n}$ is the unique asymmetric bracelet with prefix $0^{n-4}$. 
By Lemma~\ref{lem:parents}, $\parent(\alpha) \in \A(n)$.  In the first two cases of the parent rule,  $\parent(\alpha)$ will have prefix $0^{i+1}$.
If the third case applies, Lemma~\ref{lem:parents} states that $\lastone( \lastzero(\alpha))$ is an asymmetric bracelet.  Thus, $\parent(\parent(\alpha))$ is either 
$\firstone( \lastzero(\alpha))$
or $\lastone( \lastzero(\alpha))$; in each case the resulting asymmetric bracelet has prefix $0^{i+1}$.  Since either $\parent(\alpha)$ or $\parent(\parent(\alpha))$ has prefix $0^{i+1}$,  the parent rule induces a path from $\alpha$ to $\treeroot{n}$ and the height of the resulting tree is at most $2(n-4)-1$.

%=========================
%=========================

%that includes the formula (for even $n$):
%
%\[ |\A(n)| = -2^{\lfloor(n + 1)/2\rfloor} + \frac{2^{\lceil(n + 1)/2\rceil}}{4} + \frac{1}{2n} \sum_{d|n} \phi(d)2^{n/d},\]
%where $\phi$ is Euler's totient function.
%

\subsection{A successor rule} \label{sec:succ}

%Applying the framework in~\cite{binframework}, 

Each application of the parent rule $\parent(\alpha)$ in (\ref{eq:par}) corresponds to a conjugate pair.  For instance, consider the asymmetric bracelet $\alpha = 00010111\bblue{1}$.  The parent of $\alpha$ is obtained by flipping the last $1$ to obtain $00010111\bblue{0}$ (see Figure~\ref{fig:cycle9}).  The corresponding conjugate pair is $(\bblue{1}00010111, \bblue{0}00010111)$.  Let $\C(n)$ denote the set of all strings belonging to a conjugate pair in the cycle-joining tree $\cycletree_n$.  Then the following is a successor rule for an  $\OS(n)$, given $\alpha = \tt{a}_1\tt{a}_2\cdots \tt{a}_n \in \mathbf{S}(n)$:
\[f(\alpha) = \left\{ \begin{array}{ll}
         \overline{\tt{a}}_1 & \ \ \mbox{if $\tt{a}_1\tt{a}_2\cdots \tt{a}_n \in \C(n)$;}\\
         {\tt{a}_1} \  &\ \  \mbox{otherwise.}\end{array} \right.\]
%
%\edit{If $n=6$, then note that $\C(6) = \emptyset$.}
For example, if $\C(9)$ corresponds to the conjugate pairs to create the cycle-joining tree $\cycletree_9$ shown in Figure~\ref{fig:cycle9}, then the corresponding universal cycle is:
\begin{center}
\begin{tabular}{l}
$0\underline{\bblue{0}0001011
11}1001011
011001011
110011011
1\underline{\bblue{1}0001011
1}00101011
100011011
$ \\
$101011011 
100001001
110001001
010001001
100001011
001001011 
000101011,$ 
\end{tabular}
\end{center}
where the two underlined strings belong to the conjugate pair $(\bblue{1}00010111, \bblue{0}00010111)$.
In general, this rule requires exponential space to store the set $\C(n)$.  However, in some cases, it is possible to test whether a string is in $\C(n)$ without pre-computing and storing $\C(n)$. In our successor rule for an $\OS(n)$, we use Lemma~\ref{theorem:braceletTest} to avoid pre-computing and storing $\C(n)$, thereby reducing the space requirement from exponential in $n$ to linear in $n$.

%The following successor rule to construct an $\OS(n)$ tests whether or not a string $\alpha$ belongs to a conjugate pair used to join two cycles.  \red{Can we state when this will happen more clearly? Need a couple more steps here? Mos def.}

%It can be computed in $O(n)$ time using $O(n)$ space.

%==============================
\begin{result}
\noindent
{\bf Successor-rule $g$ to construct an $\OS(n)$ of length $L_n$} 

\medskip

\noindent
Let $\alpha = \tt{a}_1\tt{a}_2\cdots \tt{a}_n \in \Set(n)$ and let 
\begin{itemize}
\item $\beta_1 = 0^{n-i}\mathbf{1}\tt{a}_2\cdots \tt{a}_i$ where $i$ is the largest index of $\alpha$ such that $\tt{a}_i = 1$ (First-$1$); 
\item $\beta_2 = \tt{a}_2\tt{a}_3\cdots \tt{a}_n\bblue{1}$ (Last-$1$); 
\item  $\beta_3 =  \tt{a}_j\tt{a}_{j+1}\cdots \tt{a}_n \rred{0} 1^{j-2}$ where $j$ is the smallest index of $\alpha$ such that $\tt{a}_j = 0$ and $j > 1$ (Last-$0$).
\end{itemize}
Let

$g(\alpha) = \left\{ \begin{array}{ll}
         \overline{\tt{a}}_1 & \ \ \mbox{if $\beta_1$ and $\firstone(\beta_1)$ are in $\A(n)$;}\\
         \overline{\tt{a}}_1 &\ \  \mbox{if $\beta_2$ and $\lastone(\beta_2)$ are in $\A(n)$, and $\firstone(\beta_2)$ is not in $\A(n)$;}\\
         \overline{\tt{a}}_1 &\ \  \mbox{if $\beta_3$ and $\lastzero(\beta_3)$ are in $\A(n)$, and neither $\firstone(\beta_3)$ nor $\lastone(\beta_3)$ are in $\A(n)$;}\\
         {\tt{a}_1} \  &\ \  \mbox{otherwise.}\end{array} \right.$
\end{result}

\noindent
Starting with any string in $\alpha \in \Set(n)$, we can repeatedly apply $g(\alpha)$ to obtain the next bit in a universal cycle for $\Set(n)$. %\edit{The above successor rule starting with $\alpha = r_6 = 001011$ will repeatedly apply the fourth option, since the cycle joining tree $\cycletree_6$ consists of only the root node.  Thus six  produce $001011$.} 

\begin{exam}
The following orientable sequences for $n=6,7,8$ with lengths 6, 14, and 48 respectively are generated by applying the successor rule $g$: 
\begin{itemize}
   \item $n=6$: $001011$, 
   \item $n=7$: $00010111001011$, and
   \item $n=8$:  $000010111100101110011011100010011000101100101011$.
\end{itemize}

\vspace{-0.05in}
\end{exam}

\begin{theorem}
For $n \geq 6$, the function $g$ is a successor rule that generates an $\OS(n)$ with length $L_n$ for the set $\Set(n)$ in $O(n)$-time per bit using $O(n)$ space.
\end{theorem}
\begin{proof}
Consider $\alpha = \tt{a}_1\tt{a}_2\cdots \tt{a}_n \in \Set(n)$.  
If $\alpha$ belongs to some conjugate pair in $\cycletree_n$, then 
it must satisfy one of three possibilities stepping through the parent rule in \ref{eq:par}:
\begin{itemize}
    \item Both $\beta_1$ and $\firstone(\beta_1)$ must be in $\A(n)$.  Note, $\beta_1$ is a rotation of $\alpha$ when $\tt{a}_1=1$, where $\tt{a}_1$ corresponds to the first one in $\beta_1$.
       \item  Both $\beta_2$ and $\lastone(\beta_2)$ must both be in $\A(n)$, but additionally, $\firstone(\beta_2)$ can not be in $\A(n)$.  Note, $\beta_2$ is a rotation of $\alpha$ when $\tt{a}_1=1$, where $\tt{a}_1$ corresponds to the last one in $\beta_2$.

       \item  Both $\beta_3$ and $\lastzero(\beta_3)$ must both be in $\A(n)$, but additionally, both $\firstone(\beta_3)$ and $\lastone(\beta_3)$ can not be in $\A(n)$. 
       Note, $\beta_3$ is a rotation of $\alpha$ when $\tt{a}_1=0$, where $\tt{a}_1$ corresponds to the last zero in $\beta_3$.      
\end{itemize}
Thus, $g$ is a successor rule on $\Set(n)$ that generates a cycle of length $|\Set(n)| = L_n$. By Lemma~\ref{theorem:braceletTest}, one can determine whether a string is in $\A(n)$ in $O(n)$ time using $O(n)$ space. Since there are a constant number of tests required by each case of $g$, the corresponding $\OS(n)$ can be computed in $O(n)$-time per bit using $O(n)$ space.

    %A string $\beta$ will belong to $\A(n)$ if $\beta$ is a necklace and the necklace of $[\beta^R]$ is lexicographically larger than $\beta$.  These tests can be computed in $O(n)$ time using $O(n)$ space~\cite{Booth}.  Since there are a constant number of tests required by each case of the successor rule $g$, the corresponding $\OS(n)$ can be computed in $O(n)$ time per bit using $O(n)$ space.
\end{proof}
\begin{comment}
\begin{exam}
Consider the string $\alpha = 00101010xx \in \Set(12)$.  Then
\begin{itemize}
    \item $\beta_1  = $
    \item $\beta_2 = $
    \item $\beta_3 = $
\end{itemize}

Clearly $\beta$ and $\beta_2$ are not asymmetric bracelets. $\beta_3$ is an asymmetric bracelet; however, 
so is $\firstone(\beta_3)$.  Therefore $g(\alpha) = \tt{a}_1$.
\end{exam}
\end{comment}

%==============================================================
\section{Periodic nodes in $\cycletree_n$} \label{sec:periodic}

In this section, we present several results on periodic nodes in 
$\cycletree_n$, assuming $n \geq 6$.

\begin{lemma}\label{lem:periodic}
If a node $\alpha \in \A(n)$ from the cycle-joining tree $\cycletree_n$ is periodic, it has no children.
\end{lemma}
\begin{proof}
Let $\alpha$ be a non-root node  $\A(n)$.
We demonstrate that $\parent(\alpha)$ is aperiodic, which implies the periodic nodes in $\cycletree_n$ have no children. %Write $\alpha = 0^{j-1}1\beta$ where $\beta$ is some string, and $j\geq 1$. Since $\alpha$ is a lexicographically least rotation, the string $\beta$ does not contain $0^j$ as a substring. Consider the three possibilities for $\parent(\alpha)$. If $\firstone(\alpha)$ is in $A(n)$, then it has $0^j$ as a prefix. 
Let $j$ denote the index of the first $1$ in $\alpha$. Then $\alpha$ has prefix $0^{j-1}$ and no substring $0^j$.   Consider the three possibilities for $\parent(\alpha)$.
Suppose $\firstone(\alpha)$ is in $\A(n)$. Then it has prefix $0^j$ and is aperiodic since there is no substring $0^j$ not in the initial prefix of $0$s.  Similarly, if $\lastone(\alpha)$ is in $\A(n)$, then it has prefix $0^j$ and is aperiodic since it also has no substring $0^j$ not in the initial prefix of $0$s.
Suppose $\lastzero(\alpha)$ is in $\A(n)$ and is periodic. Then we can write $\lastzero(\alpha) = \beta^{k}$ where $k>1$ and $\beta$ is some string that contains a $1$. Either $\beta$ contains a $0$, or it does not. If $\beta$ does not contain a $0$, then $\beta = 1^i$ for some $i\geq 1$. But this implies $\alpha = 01^{n-1}$, which is not an asymmetric bracelet, a contradiction. Suppose $\beta$ contains at least one $0$. Write $\alpha = uv = yx$ where $u$, $v$, $x$, $y$ are nonempty strings such that $|u|=|x| = |\beta|$. 
 Since $\beta$ contains at least one $0$, the last $0$ in $\alpha$ must occur in $x$ and we must have $u=\beta$. Thus, one can obtain $x$ from $\beta$ by flipping a single $1$ to a $0$, which implies $x < \beta$. So we have  $xy < \beta v = uv = yx= \alpha$, which contradicts $\alpha$ being a bracelet. 
Therefore $\parent(\alpha)$ is aperiodic.
\end{proof}

\begin{lemma} \label{lem:aperiodic}
    The number of periodic nodes in $\cycletree_n$ is less than or equal to the number of aperiodic nodes in $\cycletree_n$.
\end{lemma}
%
%64-68, 77-79, 155
%
\begin{proof}
It suffices to show the existence of a 1-1 mapping $f$ from the periodic strings in $\A(n)$ to the aperiodic strings in $\A(n)$. Let $\alpha$ be periodic and in $\A(n)$. Then $\alpha = \beta^i$ for some aperiodic asymmetric bracelet $\beta$ where $i>1$. Let $p=|\beta|$. Define $f(\alpha) = 0^{p-1}1 \beta^{i-1}$. Clearly $f$ is 1-1; if $f(\alpha)=f(\alpha')$ for some periodic $\alpha'\in \A(n)$, then $f(\alpha)$ and $f(\alpha')$ share the prefix $0^{p-1}1$, which implies $\alpha=\alpha'$.
Now we prove that $f(\alpha)$ is aperiodic and is in $\A(n)$. We must have $\beta > 0^{p-1}1$, for otherwise $\alpha$ would be a symmetric bracelet. Thus, $f(\alpha)$ is an aperiodic necklace, but is not necessarily in $\A(n)$. Write $\beta = 0^k1\gamma$ where $k\geq 1$ and $\gamma$ is a non-empty string. Since $\beta$ is an aperiodic bracelet, it is an aperiodic necklace. Therefore, any nonempty proper prefix of $\beta$ cannot also be a suffix of $\beta$~\cite[Proposition 5.1.2]{Lothaire:1997}, and $\beta$ has no substring $0^{k+1}$. So $\beta^R$ must begin with a string larger than $10^k1$, and thus $1\beta = 10^k1\gamma < \beta^R1$. It follows that $f(\alpha)\in \A(n)$. 
\end{proof}

\noindent
From equation~\eqref{eqn:L}, we immediately have the following corollary.
\begin{corollary}  \label{cor:count}
    $n |\A(n)| \leq 2L_n$.
\end{corollary}

%==============================================================
%==============================================================
\section{Computing the children of a node in $\cycletree_n$} \label{sec:children}

%\red{$[\beta]$ denotes the set of all rotations of $\beta$, not a necklace or bracelet. So I added a definition of the bracelet in $\langle \beta_k\rangle$ right after the definition of $\beta_k$.}

In this section, we present an optimized way to determine the children of a node $\beta =  \tt{b}_1\tt{b}_2\cdots \tt{b}_n$ in $\cycletree_n$. We use this optimization in Section~\ref{sec:concat} to generate orientable sequences in $O(1)$-amortized time per bit. 
\newline\indent
Let $s$ and $t$ be integers such that 
$\beta$ has prefix $0^s1$ and suffix $10^{t}1$.  Let $s'$ denote the largest integer such that $0^{s'}$ is a substring of $\tt{b}_{s+1}\cdots \tt{b}_n$. Let $\beta_k$ denote $\tt{b}_1\cdots \tt{b}_{k-1}\overline{\tt{b}}_k\tt{b}_{k+1}\cdots \tt{b}_n$; it differs from $\beta$ only at index $k$. Recall that $\tilde{\beta_k}$ is the necklace in $[\beta_k]$. Let \Call{MAX}{$x,y$} denote the maximum of the integers $x$ and $y$.  Our goal is to determine the indices $k$ such that $\tilde{\beta_k}$ is in $\A(n)$ and $\parent(\tilde{\beta_k}) = \beta$. 
Consider the three cases of the parent rule $\parent$:
\begin{itemize}
    \item Suppose $\parent(\tilde{\beta_k}) = \firstone(\tilde{\beta_k}) = \beta$.  Since $\beta$ and $\beta_k$ differ only at index $k$, it must be that $k$ is the index of the first $1$ in $\beta_k$.  Thus $\tilde{\beta_k} = \beta_k$ has prefix $0^{k-1}1$ and $k \leq s$.  Since $\beta_k$ is a necklace, $k >$ \Call{MAX}{$\lfloor{s/2}\rfloor, s'$}. 
    Suppose \Call{MAX}{$\lfloor{s/2}\rfloor, s'$} $+ 1 < k < s$ and $\beta_k$ is not in $\A(n)$. Note that  $\beta_k$ is a necklace since it has a unique substring $0^{k-1}$ as a prefix.  Thus, it must be that $1\tt{b}_{k+1}\cdots \tt{b}_n \geq (1\tt{b}_{k+1}\cdots \tt{b}_n)^R$.  Since $k+1 \leq s$, this  implies that $1\tt{b}_{k+2}\cdots \tt{b}_n \geq (1\tt{b}_{k+2}\cdots \tt{b}_n)^R$ and hence $\beta_{k+1}$ is  also not in $\A(n)$.  Thus, starting from index \mbox{$k=$ \Call{MAX}{$\lfloor{s/2}\rfloor, s'$} $+ 1$} (which may or may not lead to a child), and incrementing up to $s$, we can stop testing once an index $k >$ \Call{MAX}{$\lfloor{s/2}\rfloor, s'$} $+ 1$ does not lead to a child.

   \item Suppose $\parent(\tilde{\beta_k}) = \lastone(\tilde{\beta_k}) = \beta$.  It follows from Remark~\ref{rem:lastone} that $\tilde{\beta_k} = \tt{b}_{k+1}\cdots \tt{b}_n0^{k-1}1$.  Since $\tilde{\beta_k}$ is a necklace, it must be that $k \leq \lceil{s/2}\rceil$. 
    If $k$ is the smallest index in $1, 2, \ldots , \lceil s/2 \rceil -1$ such that $\tt{b}_{k+1}\cdots \tt{b}_n0^{k-1}1$ is not in $\A(n)$, then by applying the definition of an asymmetric bracelet, it is straightforward to verify that $\tt{b}_{k+2}\cdots \tt{b}_n0^{k}1$ is also not in $\A(n)$.  Thus, starting from index $k=1$ and incrementing,  we can stop testing indices $k$ for this case once $\tt{b}_{k+1}\cdots \tt{b}_n0^{k-1}1$ is not in $\A(n)$.

    \item Suppose $\parent(\tilde{\beta_k}) = \lastzero(\tilde{\beta_k}) = \beta$. Then it must be that $k=n-1$ or $k=n-2$ from Lemma~\ref{lem:parents}. 
\end{itemize}
Based on this analysis, the function \Call{FindChildren}{$\beta$} defined in Algorithm~\ref{alg:children} will return
$\tt{c}_1\tt{c}_2\cdots \tt{c}_n$ such that $\tt{c}_k = 1$ if and only if $\tilde{\beta_k}$ is a child of $\beta$.  

\begin{algorithm}[h] 
\caption{ Determine the children of a node  $\beta = \tt{b}_1\tt{b}_2\cdots \tt{b}_n$ in $\cycletree_n$, returning $\tt{c}_1\tt{c}_2\cdots \tt{c}_n$ such that $\tt{c}_k = 1$ if and only if $\tilde{\beta_k}$ is a child of $\beta$.}
\label{alg:children}  

\small
\begin{algorithmic}[1]

\Function{FindChildren}{$\beta$}

    \State $\tt{c}_1\tt{c}_2\cdots \tt{c}_n \gets 0^n$   
    \State $s \gets$  integer such that $0^s 1$ is a prefix of $\beta$
    \State $s' \gets $  largest integer such that $0^{s'}$ is a substring of $\tt{b}_{s+1} \cdots \tt{b}_n$

    \Statex    
    
    \State \blue{$\triangleright$ FIRST 1} 
    %\State $\mathit{fail} \gets 0$    
    \For{$k$ {\bf from} $\Call{Max}{\lfloor s/2 \rfloor, s'}+1$ {\bf to} $s$}  
        \If{$0^{k-1}1\tt{b}_{k+1}\cdots \tt{b}_n \in \A(n)$} ~ $\tt{c}_k \gets 1$ 
        \ElsIf{$k > \Call{Max}{\lfloor s/2 \rfloor, s'}+1$}  \ {\bf break}  
        \EndIf
    \EndFor
    \Statex

    \State \blue{$\triangleright$ LAST 1} 
    \For{$k$ {\bf from} $1$ {\bf to} $\lceil s/2 \rceil$  }  
        \If{$\tt{b}_{k+1}\cdots \tt{b}_n0^{k-1}1  \in \A(n) $}   
            \If{ $\beta = \parent(\tt{b}_{k+1}\cdots \tt{b}_n0^{k-1}1)$}   \    $\tt{c}_k \gets 1$ 
            \EndIf
        \Else~ {\bf break}
        \EndIf
    \EndFor
    \Statex 

    \State \blue{$\triangleright$ LAST 0} 
    \If{$\tt{b}_{n-1} = 1$ {\bf and} $\tt{b}_1\cdots \tt{b}_{n-2}01  \in \A(n)$ {\bf and}  $\beta = \parent(\tt{b}_1\cdots \tt{b}_{n-2}01) $}   \  $\tt{c}_{n-1} \gets 1$
         \EndIf 

     \If{$\tt{b}_{n-1} = \tt{b}_{n-2} = 1$
    {\bf and } $\tt{b}_1\cdots \tt{b}_{n-3}011  \in \A(n)$ {\bf and } $\beta = \parent(\tt{b}_1\cdots \tt{b}_{n-3}011)$} \  $\tt{c}_{n-2} \gets 1$
        \EndIf

    %\EndIf
    %=======================================
\Statex
 \State \Return  $\tt{c}_1\cdots \tt{c}_n$

\EndFunction
\end{algorithmic}
\end{algorithm}

%=======================================

\begin{lemma} \label{lem:analyzeChildren}
The time required by calls to \Call{FindChildren}{$\beta$} summed over all $\beta \in \A(n)$ is $O(L_n)$.
\end{lemma}

\begin{proof}
Each operation in \Call{FindChildren}{} requires at most $O(n)$ work, including membership testing to $\A(n)$, and the parent function.  Consider each of the two {\bf for} loops. In the first {\bf for} loop on line 6, there are at most two membership tests to $\A(n)$ that do not detect children; for all other tests the $O(n)$ work can be assigned to the corresponding child node in $\A(n)$.  For the second {\bf for} loop starting at line 10, only one membership test to $\A(n)$ will fail; however, there may be multiple parent tests on line 12 that do not lead to a child.  In these cases, $\lastone(\tt{b}_{j+1}\cdots \tt{b}_n0^{k-1}1) = \beta$, but $\parent(\tt{b}_{j+1}\cdots \tt{b}_n0^{k-1}1) = \firstone(\tt{b}_{j+1}\cdots \tt{b}_n0^{k-1}1)$. The $O(n)$ work from each of these parent tests can be assigned uniquely to the node corresponding to the asymmetric bracelet being tested $\tt{b}_{j+1}\cdots \tt{b}_n0^{k-1}1$; each node in $\A(n)$ can receive at most one such assignment because $\lastone(\tt{b}_{j+1}\cdots \tt{b}_n0^{k-1}1) = \beta$.
Since a linear amount of work can be assigned to each $\beta \in \A(n)$, the time required by calls to \Call{FindChildren}{$\beta$} summed over all $\beta \in \A(n)$ is $O(n |\A(n)|)$. 
Thus, by Corollary~\ref{cor:count} we have our result.
\end{proof}

%=========================================
%=========================================
\section{Concatenation trees and RCL order} \label{sec:concat}
In this section, we present an algorithm based on the recent theory of concatenation trees~\cite{concat}.  It generates the orientable sequences constructed in the previous section in $O(1)$-amortized time per bit using $O(n^2)$ space.  %\edit{The definition of the concatenation tree we present in this section is a simplified version of the original definition from~\cite{concat} of our cycle-joining trees $\cycletree_n$.}

% It is based on the recent theory of concatenation trees, which are trees derived from cycle joining trees with nodes  corresponding to necklace cycles~\cite{concat}. 

%Given the cycle joining tree  $\cycletree_n$, its corresponding concatenation tree $\concattree_n$ keeps the same parent-child structure, however the label of the nodes may change, and the children are partitioned into to ordered groups: left-children and right-children.

A \defo{bifurcated ordered tree}  is  a rooted tree where each node contains two \emph{ordered} lists of children, the left-children and right-children, respectively.
 The \defo{concatenation tree $\concattree_n$} is a bifurcated ordered tree derived from $\cycletree_n$ that keeps the same parent-child structure, while the label of each node may change from $\alpha$ to another string in its necklace class $[\alpha]$, and the children are partitioned into to ordered left-children and right-children.
 The label of each non-root node becomes the string in its necklace class that differs from the label of its parent in exactly one position that we call the \defo{change index}. 
The label and change index are unique since periodic nodes appear only at the leaves (see Lemma~\ref{lem:periodic}).
 The root is labeled $\treeroot{}$ (the same label as the root of $\cycletree_n$) and assigned change index $n$.
   If a node has change index $c$, its left-children are the children with change index less than $c$, and the right-children are the children with change index greater than $c$. In both cases, the children are ordered from smallest to largest based on their change index.   As an example, the concatenation tree $\concattree_9$ in Figure~\ref{fig:bot9} is obtained from the cycle-joining tree $\cycletree_9$ illustrated in Figure~\ref{fig:cycle9}. 
\begin{figure}[ht]
    \centering
    \resizebox{5.5in}{!}{\includegraphics{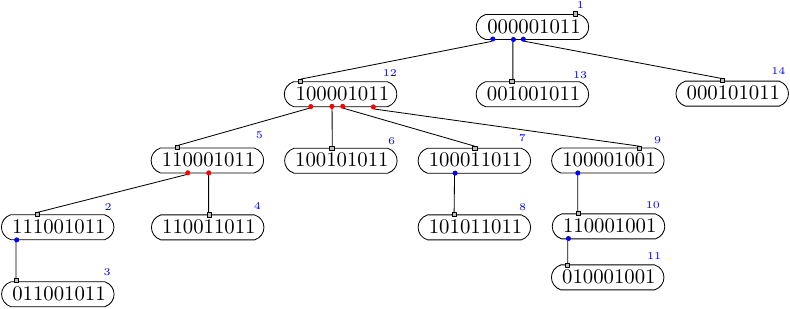}}
    \caption{The concatenation tree $\concattree_9$ derived from the cycle-joining tree $\cycletree_9$ shown in Figure~\ref{fig:cycle9}.  The small grey box on the top edge of each node indicates the change index;  
    the left-children descend from \blue{blue dots $\bullet$} and the right-children descend from \red{red dots $\bullet$ }. The small  numbers above each node indicate the order the nodes are visited in an RCL traversal.}
    \label{fig:bot9}
\end{figure}
%========================================

A \defo{right-current-left (RCL) traversal} of $\concattree_n$ starts at the root and recursively visits the right-children from first to last, followed by the current node, followed by recursively visiting the left-children from first to last.   Let \defo{$\RCL(\concattree_n)$} denote the sequence generated by traversing $\concattree_n$ in RCL order, outputting the aperiodic prefix $\ap(\alpha)$ as each node $\alpha$  is visited.
The order the nodes of $\concattree_9$ are visited by an RCL traversal is illustrated in Figure~\ref{fig:bot9}; the corresponding sequence $\RCL(\concattree_9)$ is an $\OS(9)$ of length $L_9=126$:

\begin{center}
\begin{tabular}{l}
$000001011~
111001011~
011001011~
110011011~
110001011~
100101011~
100011011
$ \\
$101011011~ 
100001001~
110001001~
010001001~
100001011~
001001011~ 
000101011.$ 
\end{tabular}
\end{center}

\noindent
In this example, each node $\alpha$ is aperiodic and hence $\ap(\alpha) = \alpha$, but this is not always the case.

%Recalling that the nodes in $\cycletree_n$ are necklace cycles and given that the alphabet is binary, the following is a simplified statement of Th

%\begin{theorem} \label{thm:main}
%Let $\cycletree$ cycle-joining tree satisfying the Chain Property.  Let $\tree_1 = \convert(\cycletree,c,\mathit{left})$.
% Then
 %\begin{itemize}
%\item $\RCL(\tree_1~)$ is a universal cycle for $\mathbf{S}_{\cycletree}$ with successor rule $\fup$, and
%\item $\RCL(\tree_2~)$ is a universal cycle for $\mathbf{S}_{\cycletree}$ with successor rule $\fdown$.
% \end{itemize}

%\vspace{-0.1in}

%\end{theorem}

The following theorem follows directly from the main result in~\cite{concat}, recalling the successor-rule $g$ defined in Section~\ref{sec:succ}.

\begin{theorem}
For $n \geq 6$, the sequence $\RCL(\concattree_n)$ is an  $\OS(n)$  of length $L_n$ that has successor-rule $g$.
\end{theorem}

\begin{comment}
\begin{center}
\begin{tabular}{r|r}
$n$ & {\bf work ratio} \\ \hline
25 & 6.746  \\
26 & 6.756  \\
27 & 6.764  \\
28 & 6.772  \\
29 & 6.779  \\
30 & 6.786  \\
31 & 6.792  \\
32 & 6.798  \\
33 & 6.804   \\
34 & 6.809   \\
35 & 6.814  \\
36 & 6.819  \\
37 & 6.824 \\
38 & 6.828  \\
39 & 6.832  \\
40 & 6.836 \\
41 & 6.839  \\
42 &   \\
\end{tabular}
\end{center}
\end{comment}

%=====================================
%=====================================
%=====================================
%\subsection{An efficient RCL traversal} \label{sec:CAT}

To avoid the exponential space required to store a concatenation tree, we demonstrate how to efficiently determine the children of a given node $\alpha=\tt{a}_1\tt{a}_2\cdots \tt{a}_n$ in $\concattree_n$.  In particular, given an index $k$, we want to determine whether or not $\alpha_k = \tt{a}_1\cdots \tt{a}_{k-1}\overline{\tt{a}}_k\tt{a}_{k+1}\cdots \tt{a}_n$
is a child of $\alpha$.  
From Lemma~\ref{lem:periodic}, if $\alpha$ is periodic, it has no children.  Otherwise, $\tilde{\alpha} = \tt{a}_s\cdots \tt{a}_n\tt{a}_1\cdots \tt{a}_{s-1}$ is a node in $\cycletree_n$  for some $1\leq s \leq n$; it is the necklace of $[\alpha]$. Thus, if $\tt{c}_1\tt{c}_2\cdots \tt{c}_n =$ \Call{FindChildren}{$\tilde{\alpha}$} (see Section~\ref{sec:children}),  $\tt{d}_1\tt{d}_2\cdots \tt{d}_n = \tt{c}_s\cdots \tt{c}_n\tt{c}_1\cdots \tt{c}_{s-1}$ is a sequence such that 
$\tt{d}_k = 1$ if and only if $\alpha_k$ is a child of $\alpha$ in $\concattree_n$.
The procedure \Call{FastRCL}{$\alpha,c$}, shown in Algorithm~\ref{algo:rcl}, applies this observation to generate $\RCL(\concattree_n)$ when initialized with $\alpha = \treeroot{}$ and $c=n$.

\begin{algorithm}[h]  
\caption{
RCL traversal of $\concattree_n$ with the initial call of \Call{FastRCL}{$\treeroot{}, n$}.  The current node $\alpha = \tt{a}_1\tt{a}_2\cdots \tt{a}_n$ has change index $c$.}
\label{algo:rcl}  

\small
\begin{algorithmic}[1]
                  
\Procedure{FastRCL}{$\alpha$, $c$}

       \State $p \gets $ period of $\alpha$ 
        \If{$p< n$}   \  \Call{Print}{$\tt{a}_1\cdots \tt{a}_p$}   \ \ \    \blue{$\triangleright$ Visit periodic node (it has no children)} 
        \Else 
               \State $s \gets $ unique index such that $\tt{a}_s\cdots \tt{a}_n\tt{a}_1\cdots \tt{a}_{s-1}$ is a necklace
                %\State $\beta = \tt{b}_1\cdots \tt{b}_n \gets \tt{a}_s\cdots \tt{a}_n\tt{a}_1\cdots \tt{a}_{s-1}$ \  \ \   \blue{$\triangleright$ The necklace of $\alpha$} 
            \State $\tt{c}_1\tt{c}_2\cdots \tt{c}_n \gets $\Call{FindChildren}{$\tt{a}_s\cdots \tt{a}_n\tt{a}_1\cdots \tt{a}_{s-1}$}
\ \ \  \blue{$\triangleright$ Determine the children indices relative to $\cycletree_n$} 
    \State $\tt{d}_1\tt{d}_2\cdots \tt{d}_n \gets \tt{c}_s\cdots \tt{c}_n\tt{c}_1\cdots \tt{c}_{s-1}$  \ \ \ \blue{$\triangleright$ Make child indices relative to $\alpha$ in $\concattree_n$}

    \Statex

    \State \blue{$\triangleright$ RCL traversal}
    \For{$i\gets c+1$ {\bf to} $n$} 
         \If{$\tt{d}_i= 1$}  \  \Call{FastRCL}{$\tt{a}_1\cdots \tt{a}_{i-1}\overline{\tt{a}}_i\tt{a}_{i+1}\cdots \tt{a}_n$, $i$}   \ \ \  \blue{$\triangleright$ Visit Right-children}     \EndIf 
    \EndFor

    \State \Call{Print}{$\tt{a}_1\cdots \tt{a}_n$}      \ \ \ \blue{$\triangleright$ Visit Current node} 
    \For{$i\gets 1$ {\bf to} $c-1$}    
         \If{$\tt{d}_i = 1$}  \  \Call{FastRCL}{$\tt{a}_1\cdots \tt{a}_{i-1}\overline{\tt{a}}_i\tt{a}_{i+1}\cdots \tt{a}_n$, $i$}  \ \ \   \blue{$\triangleright$ Visit Left-children}    \EndIf 
    \EndFor

    \EndIf

\EndProcedure

\end{algorithmic}
\end{algorithm}
%=====================

\noindent
Since $\concattree_6$ consists only of the root $r_6 = 001011$, the output of \Call{FastRCL}{$r_6$, $6$} is simply $001011$.

\begin{theorem} For $n \geq 6$, 
    \Call{FastRCL}{$\treeroot{}$, $n$} generates $\RCL(\concattree_n)$ in $O(1)$-amortized time per bit using $O(n^2)$ space.
\end{theorem}
\begin{proof} 
Each recursive call requires $O(n)$ space and from Theorem~\ref{thm:cycle}, the tree has height less than $2(n{-}4)$.  Thus, the space required by the algorithm is $O(n^2)$. 
By Lemma~\ref{lem:analyzeChildren}, the work required by all calls to \Call{FindChildren}{} is $O(L_n)$.  Ignoring these calls, there is a $O(n)$ work done at each
recursive call to \Call{RCL}{$\alpha,c$}; determining the necklace and period of a string can be computed in $O(n)$ time~\cite{Booth}.  Since there are $|\A(n)|$ nodes in $\concattree_n$,      the total work is $O(L_n)$ by applying Corollary~\ref{cor:count}.   Thus, the algorithm \Call{FastRCL}{$\treeroot{}$, $n$}, which outputs $L_n$ bits, runs in $O(1)$-amortized time per bit.

\end{proof}

%=========================================
%=========================================
\section{Extending orientable sequences}  \label{sec:extend}   

The values from the column labeled $L^*_n$ in Table~\ref{table:bounds} were found by extending an $\OS(n)$ of length $L_n$ constructed in the previous section.  Given an $\OS(n)$,  $\tt{o}_1\cdots \tt{o}_m$, the following approaches were applied to find longer $\OS(n)$s for $n\leq 20$:

\begin{enumerate}
    \item For each index $i$, apply a standard backtracking search to see whether $\tt{o}_{i}\cdots \tt{o}_m\tt{o}_1\cdots \tt{o}_{i-1}$ can be extended to a longer $\OS(n)$. We
        followed several heuristics: (a) find a maximal length extension for a given $i$, 
        and then attempt to extend starting from index $i+1$;  (b) find a maximal length extension
        over all $i$, then repeat;  (c) find the ``first'' possible extension for a given $i$, and then repeat for the next index $i+1$.  In each case, we repeat until no extension can be found for any starting index.   This approach was fairly successful for even $n$, but found shorter extensions for $n$ odd.  Steps (a) and (b) were only applied to $n$ up to $14$ before the depth of search became infeasible.
    \item  Refine the search in the previous step so the resulting $\OS(n)$ of length $m'$ has an odd number of $1$s and at most one substring $0^{n-4}$. 
    Then we can apply the recursive construction by Mitchell and Wild~\cite{MW} to generate an $\OS(n+1)$ with length $2m'$ or $2m'+1$. Then, starting from the sequences generated by recursion, we again apply the exhaustive search to find minor extensions (the depth of recursion is significantly reduced).   This approach found significantly longer extensions to obtain $\OS(n+1)$s when $n+1$ is odd. 
        
\end{enumerate}

%=========================================
%=========================================
\section{Acyclic orientable sequences} \label{sec:aperiodic}

Let $\mathcal{AOS}(n)$ denote an acyclic orientable sequence of order $n$.
If $\tt{o}_1\tt{o}_2\cdots \tt{o}_m$ is an $\OS(n)$, then it follows from our definitions that $\tt{o}_1\cdots \tt{o}_m\tt{o}_1\cdots \tt{o}_{n-1}$ is an $\mathcal{AOS}(n)$. 
As noted in~\cite{BM}, none of the $2^{\lfloor (n+1)/2 \rfloor}$ binary palindromes of length $n$ can appear as a substring in any $\mathcal{AOS}(n)$.  Thus,
a straightforward upper bound on the length of any $\mathcal{AOS}(n)$ is  
\[ \hat{U}_n = \frac{1}{2}(2^{n} - 2^{\lfloor (n+1)/2 \rfloor}) + (n-1)~~~~\cite{BM}.\]

By applying our cycle-joining based construction, we can efficiently construct an $\mathcal{AOS}(n)$ of length $L_n + (n{-}1)$. 
Previously, the only construction of $\mathcal{AOS}(n)$s recursively applied Lempel's lift~\cite{MW}, requiring exponential space.   Starting with an $\OS(n)$ found by extending a constructed sequence of length $L_n$ (see Section~\ref{sec:extend}), we  apply a computer search to extend the $\OS(n)$ to an $\mathcal{AOS}(n)$ by considering each $\tt{o}_{i}\cdots \tt{o}_m\tt{o}_1\cdots \tt{o}_{i-1}$ and attempting to extend in each direction.  This approach produced the  longest known $\mathcal{AOS}(n)$s for $n=12,13,14$, improving on the lengths discovered by Burns and Mitchell~\cite{BM} from applying a computer search.
The original data from~\cite{BM} was for $n \leq 16$; we extend the list of longest known $\mathcal{AOS}(n)$s up to $n=20$.\footnote{The resulting $\mathcal{AOS}(n)$s generated up to $n=20$ are available for download at \url{http://debruijnsequence.org/db/orientable}.}  These results are summarized in Table~\ref{tab:aperiodic}.

\noindent
%It may be possible to improve this upper bound by applying techniques similar to Lemma 2.1 in~\cite{Dai} that were applied to orientable sequences.

\begin{table}[ht]
\centering
 \begin{tabular} {r ||r  |   r  ||  r  |   r    ||   r  } 

    & \multicolumn{2}{c||}{Constructions}  &   \multicolumn{2}{c||}{Computer Search} \\ 
 $n$ &  Recursion [MW21]  &   $L_n + (n{-}1)$  &  Extended from $\OS(n)$  &   [BM93]  &  $\hat{U}_n$\\  \hline
%5 & \bblue{14}      &	-&			& \bblue{14} 	   & \bblue{14}	&  16\\
6 & \bblue{26}     &	11			&  \bblue{26} 	   & \bblue{26}	&  33\\
7 & \bblue{48}    & 20 	 		&  \bblue{48}        & \bblue{48}  	& 62 	 \\
8 & 92            & 55 				& {\bf 108}              	  & {\bf 108}            	& 127	 \\
9 & 178           & 134 			& 193               	  & {\bf 210}           	& 248	 \\
10 & 350          & 309 			& 435               	  & {\bf 440}          	& 505	 \\
11 & 692 	      & 692  			& 868               	 &  {\bf 872}          	& 1002 \\
12 & 1376         & 1541  			& {\bf 1874}        & 1860         	 & 2027  \\
13 & 2742         & 3288  			& {\bf 3732}       & 3710         	& 4044 \\
14  & 5474        & 6929  			& {\bf 7724}      & 7400            & 8141  \\
15 &10936         & 14534 			& 15432           & 15467       	& 16270  \\
16 & 21860        & 29823 			& 31560           & 31766       	& 32655	 \\
17 & 43706        & 61216 			& {\bf 63219}           & --  ~~     	& 65296	 \\
18 &  87398       & 124461 		    & {\bf 128680}          & --  ~~     	& 130833	 \\
19 &   174780     & 252842 		    & {\bf 257340}          & --  ~~     	& 	261650 \\
20 &  349544      & 509239 		    & {\bf 519212}        	  	& --   ~~    	& 	523795 \\
\end{tabular}

% n=11   821  (found from odd-parity extension)
% n=12  1874 (found from odd-parity extension)
% n=13  3732 (found from odd-parity extension)
% n=14 7724 forward only
%n=15       (odd parity extension)
%n=17 ??  63219 (found from odd-parity extension?)  63310??
%n=18 128751
%n=19 = (257320
%n=20  (519192  - not much of a gain really, as 519179 is trivial

    \caption{The lengths of the longest known $\mathcal{AOS}(n)$s found via construction and computer search for $n=6,7,\ldots, 20$.}
    \label{tab:aperiodic}
\end{table}

%=========================================
%=========================================
\section{Conclusion} \label{sec:fut}

In this paper we presented two algorithms to construct orientable sequences with asymptotically optimal length.  The first algorithm is a successor rule that outputs
each bit in $O(n)$ time using $O(n)$ space; the second algorithm generates the same sequences in $O(1)$-amortized time per bit using $O(n^2)$ space by applying a recent concatenation-tree framework~\cite{concat}.  This answers a long-standing open question by Dai, Martin, Robshaw, and Wild~\cite{Dai}.  
%parent rule has been generalized to an arbitrary alphabet like $\{C, G, A, T\}$.  The notion of orientation is especially applicable in areas of computational biology~\cite{Domaratzki:2006,Kari&etal:2005}.  
%
We conclude with the following directions for future research:

%$k$-ary root = $0^{n-2}12$.
%\begin{itemize}
%    \item decrement first non 0
%    \item increment last bit   (k-1) --> 0
%    \item increment last non-$k$
%\end{itemize}

%
\begin{enumerate}
    \item Can the lower bound of $L_n$ for $\OS(n)$s  be improved?
    \item Can small strings be inserted systematically into our constructed $\OS(n)$s to obtain longer orientable sequences?  
    %\item Can our $\OS(n)$s be used to find longer \emph{aperiodic} orientable sequences than reported in~\cite{BM}? %\footnote{Spoiler: the answer is yes, some longer aperiodic orientable sequences have just been found for $n=12,13,14$.}

    %\item \red{Can a given $\mathcal{OS}(n)$, $\tt{o}_1\tt{o}_2\cdots \tt{o}_m$, constructed from this paper be efficiently decoded? } 
\item 
    A problem closely related to efficiently generating long $\OS(n)$s is the problem of \emph{decoding} or \emph{unranking} orientable sequences. That is, given an arbitrary length-$n$ substring of an $\OS(n)$, efficiently determine where in the sequence this substring is located. 
%By efficiently, we mean that we can decode in $o(L)$ time where $L$ is the length of the entire sequence. 
There has been little to no progress in this area. Even in the well-studied area of de Bruijn sequences, %, maximal length cyclic sequences containing every fixed-length string exactly once, 
only a few efficient decoding algorithms have been discovered. %Most decoding algorithms are for specially constructed de Bruijn sequences; for example, see~\cite{Mitchell&Etzion&Paterson:1996, Tuliani:2001}. It seems hard to decode an arbitrary de Bruijn sequence. The only de Bruijn sequence whose explicit construction was discovered before its decoding algorithm is the lexicographically least de Bruijn sequence, sometimes called the \emph{Ford sequence} in the binary case, or the \emph{Granddaddy sequence} (see Knuth~\cite{knuth}). Algorithms to efficiently decode this sequence were independently discovered by Kopparty et al.~\cite{Kopparty&etal:2014} and Kociumaka et al.~\cite{Kociumaka&etal:2014}. Later, Sawada and Williams~\cite{Sawada&Williams:2017} provided a practical implementation.
\end{enumerate}

%\red{Addd link to website or include code}

%\red{The OEIS formula for A059076 is incorrect.  What to do about this?}

%\blue{I can submit a correction. I have an OEIS account. - Daniel}

%=========================
\section{Acknowledgment}
Joe Sawada (grant RGPIN-2025-03961) and Daniel Gabri\'{c} (grant RGPIN-2026-04863) gratefully acknowledge research support from the Natural Sciences and Engineering Research Council of Canada (NSERC).

%==================================
%\bibliographystyle{plainurl}
%\bibliographystyle{acm.bst}
\bibliographystyle{unsrt}

\bibliography{abbrevs,refs}

\end{document}

%% file: paper_macros.tex
\usepackage[noend]{algpseudocode}
\usepackage{algorithm}
\usepackage{times}
\usepackage{amssymb}
\usepackage{amsmath}
\usepackage{latexsym}
\usepackage{graphics}
\usepackage{url}
\usepackage{verbatim}
\usepackage{color}
\usepackage{setspace}
\usepackage{multirow}
\usepackage{lineno}
\usepackage{booktabs}
\usepackage{graphicx}
\usepackage{caption}
\usepackage[shortlabels]{enumitem}

\algnewcommand\And{\textbf{and}}
\algnewcommand\Or{\textbf{or}}
\algnewcommand\Not{\textbf{not}}
\algnewcommand\In{\textbf{in}}
\algnewcommand\Each{\textbf{each}}

\newcommand{\squishlist}{
 \begin{list}{$\bullet$}
  { \setlength{\itemsep}{0pt}
     \setlength{\parsep}{3pt}
     \setlength{\topsep}{3pt}
     \setlength{\partopsep}{0pt}
     \setlength{\leftmargin}{2.5em}
     \setlength{\labelwidth}{1em}
     \setlength{\labelsep}{0.5em} } }

\newcommand{\squishlisttwo}{
 \begin{list}{$\triangleright$}
  { \setlength{\itemsep}{0pt}
     \setlength{\parsep}{0pt}
    \setlength{\topsep}{0pt}
    \setlength{\partopsep}{0pt}
    \setlength{\leftmargin}{2em}
    \setlength{\labelwidth}{1.5em}
    \setlength{\labelsep}{0.5em} } }

\newcommand{\squishend}{
  \end{list}  }

\usepackage{listings}

\definecolor{verbgray}{gray}{0.9}

\lstnewenvironment{code}{%
  \lstset{backgroundcolor=\color{verbgray},
 frame=single,
  language=C,
  framerule=0pt,
  basicstyle=\ttfamily,
  showstringspaces=false,
  commentstyle=\color{blue}\textit,
  keywordstyle=\color{black}\bf,
  numbers=none,
  keepspaces=true,
  columns=fullflexible}}{}

\definecolor{shadecolor}{rgb}{.91, .91, .91}
\definecolor{bordercolor}{rgb}{.8, .8, .6}

\definecolor{ultramarine}{rgb}{0, 0.125, 0.376}

 \definecolor{arsenic}{rgb}{0.23, 0.27, 0.29}
 \definecolor{beige}{rgb}{0.96, 0.96, 0.86}

\definecolor{amber}{rgb}{1.0, 0.75, 0.0}
\definecolor{orange}{rgb}{1.0, 0.49, 0.0}
\definecolor{dandelion}{rgb}{0.94, 0.88, 0.19}

  \definecolor{indiagreen}{rgb}{0.07, 0.53, 0.03}
  \definecolor{huntergreen}{rgb}{0.21, 0.37, 0.23}

\newcommand{\blue}[1] {\textcolor{blue}{#1}}
\newcommand{\red}[1] {\textcolor{red}{#1}}

\newcommand{\bblue}[1] {{\bf \textcolor{blue}{#1}}}
\newcommand{\rred}[1] {{\bf \textcolor{red}{#1}}}

\newcommand{\defo}[1] {\emph{\textcolor{blue}{#1}}}

\definecolor{shadecolor}{rgb}{.9, .9, .9}

 \usepackage{framed}

    {\endMakeFramed}

    \newenvironment{frshaded*}{%
    \MakeFramed {\advance\hsize-\width \FrameRestore}}%
    {\endMakeFramed}
    
    \newcounter{examplecounter}
\newenvironment{exam}{
 \begin{frshaded*}
    \refstepcounter{examplecounter}%
    \noindent
  \textbf{Example \arabic{examplecounter}}%
  \quad
}{%
\end{frshaded*}
}

{\endMakeFramed}

\newenvironment{frshaded2*}{%
    \MakeFramed {\advance\hsize-\width \FrameRestore}}%
    {\endMakeFramed}

\newenvironment{result}{
 \begin{frshaded2*}
}{%
\end{frshaded2*}

}

{\endMakeFramed}

\newenvironment{frshaded3*}{%
    \MakeFramed {\advance\hsize-\width \FrameRestore}}%
    {\endMakeFramed}

\graphicspath{{graphics/}}